\documentclass[11pt]{article}
\usepackage[T1]{fontenc}
\usepackage[utf8]{inputenc}
\usepackage[margin=2.7cm]{geometry}
\usepackage{amsmath,amssymb,amsthm}
\usepackage{array}
\usepackage{booktabs}
\usepackage{float}
\usepackage{placeins}
\usepackage{graphicx}
\usepackage{natbib}
\usepackage[colorlinks=true,linkcolor=blue,citecolor=blue,urlcolor=blue]{hyperref}
\hypersetup{pdftitle={Every pooling rule has its world: matching probability
combination rules to situations and stakes},
pdfauthor={Tanel Tammet, Priit J\"arv, Dirk Draheim},
  bookmarksdepth=2
}

\newtheorem{proposition}{Proposition}

\title{Every pooling rule has its world:\\
matching probability combination rules to situations and stakes}

\author{Tanel Tammet, Priit J\"arv, Dirk Draheim\\
Tallinn University of Technology\\
\texttt{tanel.tammet@taltech.ee}
}

\date{August 2026}

\begin{document}
\maketitle

\begin{abstract}
Systems often need to combine two numerical assessments of the same yes/no
question.  The appropriate formula depends on what the numbers represent and
on how the sources are related.  Averaging is correct when one of several
alternative interpretations applies; multiplying odds is correct when
probability reports are based on conditionally independent evidence and a
common prior; and probabilities of alternative successful derivations require
their dependence or shared evidence to be taken into account.

We state the assumptions behind several common combination rules and derive
the corresponding combined probabilities.  Two groups of Monte Carlo
experiments address different questions.  First, controlled generating
mechanisms verify that the derived rule recovers the correct probability in
the situations for which its assumptions hold.  Second, the same mechanisms
measure the consequences of using a mismatched rule, using logarithmic score
and threshold decisions with different costs.  Distinct pooling rules can
produce the same binary decision at threshold $1/2$ while assigning
substantially different probabilities, so binary accuracy alone can conceal
important differences.  We also give probabilistic interpretations of
conflicting-evidence rules and show that, for overlapping derivations,
retaining the identities of shared uncertain premises permits direct
calculation of the probability that at least one derivation is available.
Pairwise combination of proof probabilities loses information when there are
three or more derivations.
\end{abstract}

\section{Introduction}

Two numerical assessments of the same yes/no question may come from ensemble members,
retrieved passages, sampled solutions, judge models, or a generator and a critic. The
correct formula for combining them depends on what the numbers represent and on how the
sources are related. If one of two alternative interpretations
applies, the combined probability is a weighted average. If two probability reports are
based on conditionally independent evidence and the same prior, their odds multiply. If
the numbers are success probabilities of two attempts, the combined probability depends
on how the attempts overlap. We state these assumptions for common combination rules and
measure the loss when a rule's assumptions do not hold.

Linear, geometric, and multiplicative pooling have long been studied in statistics and
formal epistemology \citep{stone1961,genestzidek1986,dietrichlist2016}. Forecast
aggregation also studies log-odds averaging \citep{satopaa2014,baron2014} and
performance-based weighting \citep{cooke1991}. Expert systems, probabilistic logic
programs, and evidence theory use additional rules for combining confidence values and
proofs
\citep{shortliffebuchanan1975,buchananshortliffe1984,heckerman1986,deraedt2007,fierens2015,dempster1967,shafer1976}.
Related machine-learning applications include ensembles
\citep{lakshminarayanan2017,ovadia2019}, mixtures of experts
\citep{hinton2002,jacobs1991}, retrieval marginalization \citep{lewis2020},
sampled-answer voting \citep{wang2023,xiong2024}, and model panels \citep{verga2024}.
Automated commonsense reasoning propagates confidences along proofs and cumulates
parallel derivations \citep{tammet2021}; its successor in the GK reasoner records shared
uncertain premises when combining retained proofs (system lineage \citealp{tammet2022};
current mechanism \citealp{tammet2026}).

The axiomatic literature characterizes these pools by the properties they satisfy
\citep{aczelwagner1980,mcconway1981,genest1984,dietrichlist2016}. These results
characterize formal properties of pooling rules. They do not determine which rule is
correct for a given way of producing the reports, or how costly a mismatch is. For
reports 0.3 and 0.8 the formulas considered here return values from
0.55 to 0.86. Four of them produce the same decision at threshold $1/2$. Binary accuracy
at that threshold cannot distinguish them.

\begin{sloppypar}
Studies of language-model confidence mainly compare voting, averaged confidence, and
confidence-weighted voting, evaluated by calibration and error-detection performance
\citep{xiong2024,lyu2024,rivera2024,verga2024}. They usually select the rule from
benchmark results and do not model dependence among the reports. We study how the
meaning and dependence of the reports determine the appropriate rule.
\end{sloppypar}

The paper makes four main contributions:
\begin{enumerate}
\item For common probability and confidence-combination rules, we state what
their inputs represent and derive generating conditions under which their
output is the correct conditional probability.

\item We separate two experimental questions.  Controlled Monte Carlo worlds
first test whether each rule matches the probability implied by its stated
assumptions.  The same worlds are then used to measure the cost of a mismatch
with logarithmic score and threshold decisions under different costs.

\item We show that numerically different pooling rules can nevertheless induce
the same binary decision at threshold $1/2$.  Consequently, binary accuracy at
one threshold can fail to distinguish materially different probability
estimates.  We also give explicit generative interpretations for several
rules for combining supporting and opposing evidence.

\item For derivations that share uncertain premises, we show that retaining
the identities of those premises permits direct calculation of the
proof-union probability, while reducing proofs to pairwise scalar
probabilities loses information for three or more derivations.
\end{enumerate}

The two experimental parts answer different questions.  Section~3 asks a
semantic question: given a specified mechanism that generates the reports and
the outcome, which pooling rule recovers the resulting conditional
probability?  Section~4 asks a practical question: how much does it matter if
a different rule is used?  It reuses the same generating mechanisms and
measures the resulting probability error through proper scoring and
cost-sensitive threshold decisions.  Thus the second experiment does not
introduce another notion of correctness; it measures the consequences of the
mismatches established in the first.

Section~\ref{sec:pools} defines the rules and their input types. Section~\ref{sec:corr}
derives, for each rule, conditions under which it is correct. Section~\ref{sec:betting}
measures the loss caused by mismatched rules.

Each experiment is a self-contained Python program; the complete suite
is publicly available in release \texttt{v2.0} (commit \texttt{40f2cac}) of
\url{https://github.com/tammet/poolingworlds} (Appendix~A), and the
essential code is reproduced in the appendix. The experiments use two sources, yes/no
questions, and fully specified generating mechanisms. Section~\ref{sec:summary}
discusses extensions to more than two sources and the estimation of source models and
weights.

\section{The pooling functions}
\label{sec:pools}

\subsection{Setting and scales}

Two sources report numbers $p_1$ and $p_2$ in the interval $[0,1]$ about the same yes/no
question. A pooling function $f(p_1,p_2)$ returns a combined number in $[0,1]$. Most
methods use two numerical inputs; the GK calculation of Section~\ref{sec:gk} also uses
the sets of uncertain premises behind the two proof values. The paper uses a probability
scale and a confidence scale.

On the \emph{probability scale}, the report is a probability that the answer is yes; $1/2$
means maximal uncertainty. The average, geometric, multiplicative and log-odds rules are
defined on this scale.

On the \emph{confidence scale}, used in certainty-factor systems and in confidence-annotated
logic, the report measures the strength and direction of the evidence: $1/2$ means no
information, values above $1/2$ support the claim, values below oppose it. A confidence $c$
maps to a signed strength $s=2c-1$ in $[-1,1]$. The MYCIN rule, the ProbLog conflict rule and
the naive difference are stated in terms of strengths; Dempster--Shafer combination
reads simple-support masses, a separate input type (Section~\ref{sec:types}).
When the two sources have opposite signs we write $a$ for the strength of the
supporting evidence and $b$ for the strength of the opposing evidence, both in $[0,1]$.

Two probability reports about one outcome can describe alternative cases or provide
evidence about the same case. In the first case, only one of the alternatives applies,
and the overall probability is obtained by averaging over the cases: an unobserved
variable $Z$ selects the case, report $p_z$ means $P(Y=1\mid Z=z)$, and
\[
P(Y=1)=\sum_z P(Z=z)\,P(Y=1\mid Z=z),
\]
so averaging is correct when the uncertainty is about which reading, passage, expert, or
model applies. In the second case, the correct combination depends on whether the
evidence is independent, shared, or attached to separate component events. The
applicable rule depends on the input type and the source relation.

\subsection{The types of the inputs}
\label{sec:types}

The rules use five input types:
\begin{align*}
p_i &= P(Y=1\mid E_i) &&\text{the probability of the outcome after observing evidence } E_i;\\
r_i &= P(A_i) &&\text{the probability of one event that can make the outcome occur};\\
s_i &\in[-1,1] &&\text{a signed strength of supporting or opposing evidence};\\
m_i &\in[0,1] &&\text{mass assigned to support or opposition, the rest to ignorance};\\
S_i & &&\text{the set of uncertain premises used by proof } i.
\end{align*}
The table below groups the rules by the meaning of their inputs and by the quantity
being computed. In the narrow sense, ``opinion pooling'' refers only to rules that
combine posterior probability reports; we use \emph{probability and
confidence-combination rules} as the covering term.

\begin{center}
\footnotesize
\setlength{\tabcolsep}{3pt}
\begin{tabular}{@{}l>{\raggedright\arraybackslash}p{4.4cm}l@{}}
\toprule
input & quantity to compute & rules\\
\midrule
posterior reports about one $Y$ & probability of $Y$ & average, geometric, upco, log-odds\\
probabilities of events $A$, $B$ & probability of $A\lor B$ & max, noisy-or, bounded sum, cumulation\\
strengths for and against a claim & combined evidence strength & MYCIN, ProbLog rule, naive difference\\
masses on true, false, ignorance & lower and upper probabilities; decision probability & Dempster--Shafer\\
uncertain premises behind proofs & at least one proof available & GK calculation\\
\bottomrule
\end{tabular}
\end{center}

The following rules are considered. Combined values for the input pair $(0.3,\,0.8)$ are given
in parentheses to show the spread.

\subsection{Average and weighted average (0.55)}

\emph{Input:} probabilities under alternative cases, together with the probability of
each case. \emph{Target:} the overall probability after averaging over the cases (the
marginal probability).

\begin{equation*}
\mathrm{avg}(p_1,p_2)=\frac{p_1+p_2}{2},
\qquad
\mathrm{wavg}(p_1,p_2)=w\,p_1+(1-w)\,p_2 .
\end{equation*}
The weighted average is the standard linear pool \citep{stone1961}. Its value lies in
the interval spanned by the two reports: if both reports exceed $1/2$, the average does
not exceed the larger report.\footnote{Weighted averages are characterized by marginal
dependence and unanimity preservation under the conditions of \citet{mcconway1981} and
\citet{aczelwagner1980}.}

\subsection{Maximum (0.8), probabilistic sum (0.86), and bounded sum (1.0)}
\label{sec:tconorms}

\emph{Input:} probabilities of two component events. \emph{Target:} the probability
that at least one occurs.

\begin{equation*}
\max(p_1,p_2),
\qquad
\mathrm{ps}(p_1,p_2)=p_1+p_2-p_1p_2,
\qquad
\min(1,\,p_1+p_2).
\end{equation*}
These three rules apply when the reports are the probabilities of two \emph{events} and the
question is whether at least one of them occurs. The dependence relation determines the
union probability: the maximum applies when the events are maximally positively dependent
(one contains the other), the probabilistic sum (called noisy-or in the Bayesian-network
literature) when they are independent, and the bounded sum when the union has the largest
probability compatible with the two reports (the upper Fr\'echet bound). Given only the two marginal probabilities, the union probability is not determined; it can
lie anywhere in the interval
\begin{equation*}
\max(p_1,p_2)\;\le\;P(A\ \text{or}\ B)\;\le\;\min(1,\,p_1+p_2),
\end{equation*}
the Fr\'echet bounds, studied for logical inference by \citet{hailperin1986} and
\citet{nilsson1986}. All three rules are commutative and associative.

\subsection{Multiplicative pooling, upco, product of experts (0.63)}
\label{sec:upco}

\emph{Input:} two calibrated posterior reports about one claim, computed from a common
prior. \emph{Target:} the posterior probability given both observations.

\begin{equation*}
\mathrm{upco}(p_1,p_2)=\frac{p_1p_2}{p_1p_2+(1-p_1)(1-p_2)} .
\end{equation*}
The reports are multiplied and renormalized; on the odds scale, writing $o=p/(1-p)$, the
combined odds are $o_1 o_2$. The formula is correct under three assumptions: both
reports are calibrated posterior probabilities, both use the same prior, and their
observations are conditionally independent given either outcome. With prior
$1/2$, the posterior given both observations is $\mathrm{upco}(p_1,p_2)$: likelihood
ratios of independent observations multiply. For a prior $\pi$ different from $1/2$ the correct combination
divides the doubled prior odds back out once:
\begin{equation*}
\text{oddspool}(p_1,p_2;\pi):\qquad
\text{combined odds}=o_1\, o_2\,\frac{1-\pi}{\pi}.
\end{equation*}
Using upco at a prior different from $1/2$ counts the prior twice. The same formula
occurs as multiplicative opinion pooling \citep{dietrichlist2016}, as upco (``updating
on the credences of others'') \citep{easwaran2016,pettigrewweisberg}, and as the product
of experts \citep{hinton2002}; prior-corrected variants occur in Bayesian group belief
\citep{dietrich2010} and the Bayesian committee machine \citep{tresp2000}. Under the independent-evidence
condition, two equal reports produce a result farther from $1/2$. Under other source
models the same operation can overcount the evidence.\footnote{The rule is
commutative and associative on $(0,1)^2$; at total conflict, $(0,1)$ or $(1,0)$, the
normalized product is $0/0$, and an implementation that returns $1/2$ there applies a
convention, not the rule (the same holds for the geometric pool).}

\subsection{Geometric pooling and the weighted log-odds family (0.57 for equal weights)}
\label{sec:geo}

\emph{Input:} two posterior reports about one claim. \emph{Target:} the combined
posterior in the cases described below.

\begin{equation*}
\mathrm{geo}(p_1,p_2)=\frac{(p_1p_2)^{1/2}}{(p_1p_2)^{1/2}+\bigl((1-p_1)(1-p_2)\bigr)^{1/2}},
\end{equation*}
that is, the combined odds are $(o_1o_2)^{1/2}$. More generally, converting each report to
log-odds $\ell=\ln o$ and forming a weighted sum gives the family
\begin{equation*}
\text{pooled log-odds}=w_1\ell_1+w_2\ell_2 .
\end{equation*}
Geometric pooling is the member $(1/2,1/2)$; upco is $(1,1)$. The rule is exact when
both reported log-odds have the same multiplicative distortion, when both reports repeat
the same evidence, or when two equal-sized samples are averaged
(Sections~\ref{sec:geocorr} and~\ref{sec:knob}).

Throughout the paper $\alpha$ denotes the symmetric per-source coefficient,
$w_1=w_2=\alpha$. Forecast aggregation often writes the same family as pooled log-odds
$=\gamma\,(\ell_1+\ell_2)/2$; hence $\gamma=2\alpha$, with $\gamma=1$ the geometric
pool and $\gamma=2$ upco. Multiplying the average log-odds by $\gamma>1$ (extremizing)
improved forecasting-tournament scores in the cited studies \citep{baron2014,satopaa2014}.

Section~\ref{sec:knob} evaluates the symmetric coefficient $\alpha$ under four source
models: $\alpha=1/k$ (extremization factor $\gamma=2/k$) is the exact correction when
every report exaggerates or understates its log-odds by a factor $k$, and per-source
weights $1/k_i$ handle sources with different distortion factors.

\subsection{The MYCIN certainty-factor rule (0.71 on the confidence scale)}
\label{sec:mycin}

\emph{Input:} signed evidence strengths for and against a claim. \emph{Target:} the
remaining support after the opposing evidence is taken into account.

We use the revised EMYCIN certainty-factor combination:\footnote{The revised rule is
the version used by later MYCIN-family systems, documented in
\citet[ch.~10]{buchananshortliffe1984} and \citet{vanmelle1980}. The original rule
\citep{shortliffebuchanan1975} differed in the conflict case, using plain addition
$x+y$ for opposite signs; a single moderate disconfirmation could then cancel an
arbitrarily long run of confirmations, and van Melle's EMYCIN replaced the conflict
case. Implementations that return the neutral value $1/2$ at total conflict adopt a
boundary convention outside the rule.}
\begin{align*}
\text{same sign, positive:}&\quad x+y-xy\\
\text{same sign, negative:}&\quad x+y+xy\\
\text{opposite signs:}&\quad \frac{x+y}{1-\min(|x|,|y|)} .
\end{align*}
On the confidence scale
with supporting strength $a$ and opposing strength $b$, $a>b$, the conflict case reads
$(a-b)/(1-b)$. The rule is commutative and, except at total conflict, associative;
at total conflict (strengths $1$ and $-1$; confidences $0$ and $1$) the additive
representation below diverges and the rule is undefined.
Associativity follows from an
additive representation: under the transformation
\begin{equation*}
\lambda(x)=-\ln(1-x)\quad\text{for }x\ge 0,\qquad \lambda(-x)=-\lambda(x),
\end{equation*}
the transformed values add:
$\lambda(\mathrm{combine}(x,y))=\lambda(x)+\lambda(y)$.
Section~\ref{sec:mycincorr} interprets $\lambda$ as an evidence-arrival rate: rates add
on the same side, and the conflict rule is the probability that at least one supporting
item survives item-level refutation. \citet{heckerman1986} showed that certainty
factors are consistent with probability theory when read as monotone functions of
likelihood ratios, with the parallel combination sound only for evidence independent
given the hypothesis and given its negation \citep[see also][]{heckermanshortliffe1992}.

\subsection{The ProbLog conflict rule (0.56) and the naive difference (0.50)}
\label{sec:problog}

\emph{Input:} signed evidence strengths for and against a claim. \emph{Target:} the
probability that the supporting derivation holds and the opposing one does not.

In ProbLog, independent probabilistic derivations of the same fact are combined by
their union probability, the probabilistic sum. In the support-with-veto ProbLog program
used here (four lines, reproduced in Appendix~A), the query succeeds when the supporting
derivation is present and the opposing derivation is absent; under independence, its
probability is $a(1-b)$ \citep{deraedt2007,fierens2015}. We use ``the ProbLog rule''
for this construction; different programs can represent different relations between
supporting and opposing evidence. The naive difference $a-b$, used for conflicting evidence in the CONFER
confidence framework for automated reasoning \citep{tammet2021}, subtracts the strengths
directly. Both rules agree with MYCIN and Dempster--Shafer on the agreement case
(probabilistic sum) and differ on the conflict case; neither is associative.

For confidence inputs $c_1,c_2$, we use the following piecewise form. Writing
$s_i=2c_i-1$ for the signed strengths, the combined strength is
\begin{equation*}
z=
\begin{cases}
|s_1|+|s_2|-|s_1|\,|s_2| & \text{if } s_1\ge 0 \text{ and } s_2\ge 0,\\[2pt]
-\bigl(|s_1|+|s_2|-|s_1|\,|s_2|\bigr) & \text{if } s_1\le 0 \text{ and } s_2\le 0,\\[2pt]
a(1-b) & \text{if the signs are strictly opposite,}
\end{cases}
\end{equation*}
and the combined confidence is $(1+z)/2$; in the conflict case $a$ is the magnitude of the
supporting strength and $b$ of the opposing strength, and the resulting signed strength
$a(1-b)$ is nonnegative. Confidence $1/2$ is neutral: a zero strength falls into a same-sign
branch, which returns the other input unchanged. Against an opposing input, the
supporting branch tends to confidence $1/2$ as its support tends to zero; at the exact
neutral point, the output equals the opposing input. The operator is therefore discontinuous
at the branch boundary --- the jump visible in Figure~\ref{fig:rulesconf}.

\subsection{Overlapping derivations: estimated overlap and shared premises}
\label{sec:overlap}

Two derivations of the same conclusion may use some of the same uncertain premises. With
complete overlap, the union probability is the larger derivation probability; with no
overlap and independent premises, it is the probabilistic sum. Cumulation selects an
intermediate value from an estimated degree of overlap; the GK calculation uses the
identities and probabilities of the shared premises.

\subsubsection{Cumulation: interpolation between complete overlap and independence}
\label{sec:cumul}

\emph{Input:} two derivation probabilities and an estimated degree of independence.
\emph{Target:} the probability that at least one derivation holds.

The cumulation rule of \citet{tammet2021} is
\begin{equation*}
\mathrm{cumul}(c_1,c_2;a)=\max(c_1,c_2)+a\bigl(\mathrm{ps}(c_1,c_2)-\max(c_1,c_2)\bigr),
\qquad a\in[0,1],
\end{equation*}
where $\mathrm{ps}(c_1,c_2)=c_1+c_2-c_1c_2$ is the probabilistic sum. The supplied parameter $a$ estimates
independence: $a=1$ represents independent derivations and $a=0$ complete overlap. The
formula is exact under the following mixture model. With probability $1-a$, both
derivation events are generated from the same uniform variable. With probability $a$,
they are generated independently.\footnote{Equivalently, cumulation is the union probability under a mixture
of the comonotone and independent copulas \citep{nelsen2006}.}
Section~\ref{sec:cumulcorr} builds this situation and shows that cumulation is exact for
it. When the overlap comes from specific shared premises, the estimated parameter is
generally approximate (Section~\ref{sec:cumulcorr} quantifies the error); Section~\ref{sec:gk} calculates the
overlap from retained event identities. For fixed $0<a<1$ the rule is non-associative;
for three or more derivations, an order-independent calculation requires assumptions
about their joint dependence.

\subsubsection{GK: combining proofs with shared uncertain premises}
\label{sec:gk}

\emph{Input:} the uncertain ground premises used by each retained proof, together with
their probabilities; GK calls these premises activation events. \emph{Target:} the
probability that at least one proof is available.

GK records which uncertain ground premises occur in each retained proof; shared premises
are counted once when the probability of the union of the proofs is calculated
\citep{tammet2026}. Under GK's activation model, each
uncertain ground clause instance is an independent \emph{activation
event}. A retained proof is available when all
events in its support set occur, so a proof with event set $S_i$ has availability
$c_i=\prod_{e\in S_i}p_e$. The value is the probability that the retained proof is
available under the activation model; a posterior truth probability for the conclusion
requires a separate model.

A three-event case shows the calculation. Suppose proof 1 requires the events $\{u,v\}$ and
proof 2 the events $\{v,w\}$. Then $c_1=p_up_v$, $c_2=p_vp_w$, and the shared part is the
single event $v$ with probability $s=p_v$. Both proofs are available exactly when all three
events occur, with probability $p_up_vp_w=c_1c_2/s$, so the probability that at least one
proof is available is $c_1+c_2-c_1c_2/s$.

In general, for two retained proofs of the same answer, let
$s=\prod_{e\in S_1\cap S_2}p_e$, the product over the \emph{shared} events ($s=1$ when
they share none). A shared event is counted once in the conjunction of the two proofs, so
$P(\text{both available})=c_1c_2/s$, and the proof-union probability reduces to the
closed form
\begin{equation*}
\mathrm{gk}(c_1,c_2;s)=c_1+c_2-\frac{c_1c_2}{s}.
\end{equation*}
Since each proof contains all shared events, $c_1\le s$ and $c_2\le s$. The method uses both the proof
availabilities and the retained event sets: disjoint event sets give the probabilistic
sum, containment of one set in the other gives the maximum, and for partial overlap the
shared event identities determine the probability that both proofs are available. The
model assumes that distinct activation events are independent.

\paragraph{Scope.}
The calculation combines proofs supporting the same answer; supporting and opposing
proofs are never placed into one union. When supporting and opposing proofs interact, GK
uses a separate calculation described in \citet{tammet2026}.

\paragraph{Exactness.}
The result is exact for the retained proof set under the activation model, which treats
the supplied $p_e$ as activation probabilities, when every event has a fully
instantiated identifier, the event sets can be recovered from the proof histories, and
the exact inclusion--exclusion limits are not exceeded.\footnote{GK treats separately entered
uncertain clauses as separate events. The current implementation performs exact
inclusion--exclusion within stated proof-set and event limits; see \citet{tammet2026}.} The result equals the full query probability only when the retained
proofs include all relevant minimal proofs; omitting a proof can only reduce the union
probability.

\paragraph{More than two proofs.}
The two-proof formula cannot be folded: after two proofs are reduced to one number,
the event identities needed to determine the overlap with a third proof are lost, and
Section~\ref{sec:gkcorr} gives a chain-sharing example in which every pairwise fold
gives an incorrect value.
For more than two proofs, GK applies inclusion--exclusion to all retained proofs
together, calculating the probability of each intersection from the union of the
premises used by those proofs; the result is independent of proof order. If event identities are unavailable, the overlap must again be
represented by an estimated parameter (Section~\ref{sec:cumul}).

\paragraph{Relation to established methods.}
Related probabilistic-logic and weighted-model-counting methods also evaluate unions of
overlapping proofs \citep{green2007,kimmig2017,fierens2015}. GK reconstructs the
uncertain ground premises from retained first-order proof histories and applies bounded
inclusion--exclusion without grounding the full knowledge base \citep{tammet2026}. In
GK, this calculation replaces the earlier CONFER cumulation rule of
Section~\ref{sec:cumul} \citep{tammet2021,tammet2022}. Direct inclusion--exclusion is
exponential in the number of retained proofs, so larger cases require compilation or
approximation.

\subsection{Dempster--Shafer combination for a yes/no question (bel 0.74, BetP 0.83)}
\label{sec:ds}

\emph{Input:} one source assigns mass to support and the rest to ignorance; the other
assigns mass to opposition and the rest to ignorance. \emph{Target:} lower and upper
probabilities, and a decision probability that assigns half of the remaining ignorance
to each outcome ($\mathrm{BetP}$, the pignistic probability).

Evidence theory assigns masses to sets of hypotheses. In the binary simple-support case one
source assigns mass $m_1$ to ``the claim is true'' and the rest to ignorance, the other assigns
$m_2$ to ``the claim is false'' and the rest to ignorance. Dempster's rule discards the
contradictory combination $m_1m_2$ and renormalizes \citep{dempster1967,shafer1976}.
After conditioning on non-conflict, three cases remain: only the supporting mass applies
(probability $m_1(1-m_2)$), only the opposing mass applies ($m_2(1-m_1)$), and neither
applies ($(1-m_1)(1-m_2)$, unresolved ignorance). After
renormalization by $1-m_1m_2$, $\mathrm{bel}$ is the probability of the first case,
$\mathrm{pl}$ adds the ignorance case, and $\mathrm{BetP}$ adds half of it:
\begin{align*}
\mathrm{bel}&=\frac{m_1(1-m_2)}{1-m_1m_2} &&\text{(lower probability)}\\
\mathrm{pl}&=1-\frac{m_2(1-m_1)}{1-m_1m_2} &&\text{(upper probability)}\\
\mathrm{BetP}&=\mathrm{bel}+\frac{\mathrm{pl}-\mathrm{bel}}{2}
&&\text{(pignistic probability, \citealp{smetskennes1994}).}
\end{align*}
The agreement case (both masses on the same side) is again the probabilistic sum. On the
confidence scale the conflict value $\mathrm{bel}$ equals the ProbLog value $a(1-b)$
renormalized by $1-ab$. On the same numerical strengths, the four conflict rules give,
for $(a,b)=(0.8,\,0.3)$:
\begin{equation*}
\text{naive } 0.50 \;<\; \text{ProbLog } 0.56 \;<\; \text{MYCIN } 0.71
\;\sim\; \text{DS bel } 0.74 \;<\; \text{DS BetP } 0.83 .
\end{equation*}
Dempster's normalization discards the conflicting mass, which has been criticized
\citep{zadeh1984}. In the sampling model of Section~\ref{sec:dscorr}, the normalization
is conditioning on the two opposing sources not both having decisive evidence. When
$m_1$ and $m_2$ are both large, little mass remains after conditioning; at total
conflict ($m_1m_2=1$) the rule is undefined, and the redraw in the sampling model never
terminates. Aggregation frameworks based on sets of probabilities avoid the
renormalization \citep{benavoli2010}. Section~\ref{sec:dscorr} gives a probability model under which
these quantities are exact. Smets' transferable-belief
model uses $\mathrm{BetP}$ for decisions; other belief-function decision criteria use the
lower or upper expectations \citep{denoeux2019}. World 9 of Section~\ref{sec:worlds}
compares these quantities.

\subsection{The rules side by side}
\label{sec:sidebyside}

Table~\ref{tab:semantics} collects the rules in one place: what the inputs represent,
when the rule is correct (established in Section~\ref{sec:corr}), and its main
limitation or mismatch effect. Earlier work compares pooling rules by axiomatic
properties \citep{stewartojea2018}.

\begin{table}[H]
\centering
\footnotesize
\setlength{\tabcolsep}{4pt}
\begin{tabular*}{\textwidth}{@{\extracolsep{\fill}}l
  >{\raggedright\arraybackslash}p{4.1cm}
  >{\raggedright\arraybackslash}p{3.3cm}
  >{\raggedright\arraybackslash}p{4.2cm}@{}}
\toprule
rule & what the inputs represent & correct when (Sec.~\ref{sec:corr}) & main limitation or mismatch effect\\
\midrule
average & probabilities under alternative cases & mixture over alternative cases & understates reinforcement from independent evidence\\
weighted avg. & as above, unequal weights & unequal case probabilities & error from misspecified weights\\
maximum & probabilities of two nested events & one event contains the other & omits the contribution of non-contained events\\
prob.\ sum & probabilities of two independent events & independent events & overstates positively dependent posterior reports\\
bounded sum & probabilities of two events, dependence unknown & (upper Fr\'echet bound) & upper bound\\
upco & posteriors from separate evidence, prior $1/2$ & independent evidence, shared prior & double-counts shared evidence\\
odds pool & as above, prior $\pi$ & independent evidence, any prior & needs the prior\\
geometric & posterior reports or densities & identical evidence; doubled logits; averaged samples & understates independent private evidence\\
MYCIN & signed evidence strengths & item-level refutation & models evidence survival\\
ProbLog & signed evidence strengths & independent opposition & discontinuous at $1/2$\\
cumulation & two proof probabilities and an estimated probability of independence & all-or-nothing sharing & approximate for item-level overlap\\
GK calculation & proofs with identified shared uncertain premises & proofs of the same answer, independent premises & requires ground premise identities and bounded exact calculation; premise sets needed beyond two proofs\\
DS bel/pl/BetP & mass assigned to support, opposition, and ignorance & both sources cannot be decisively correct; BetP assigns half of the remaining ignorance to each outcome & bel/pl are bounds, not point decision probabilities\\
\bottomrule
\end{tabular*}
\caption{The rules at a glance. The conditions are established in
Section~\ref{sec:corr}; the last column names the main limitation or the typical
mismatch effect, measured in Section~\ref{sec:betting}.}
\label{tab:semantics}
\end{table}

Figures~\ref{fig:rulesprob} and~\ref{fig:rulesconf} compare the outputs of the rules on
the same numerical inputs: the second input is held fixed and the first varies from 0
to 1. In
Figure~\ref{fig:rulesprob}, the maximum and probabilistic-sum curves are bounded
below by the fixed input, while the average, the geometric pool and upco all cross
$1/2$ at the same point $x=1-y$, the decision equivalence proved as
Proposition~\ref{prop:half} below. In
Figure~\ref{fig:rulesconf}, the MYCIN rule is continuous across the branch change at
$x=1/2$ and the ProbLog rule jumps there when the fixed input opposes
(Section~\ref{sec:problog}); upco is plotted in both figures for comparison.

\begin{figure}[H]
\centering
\includegraphics[width=\textwidth]{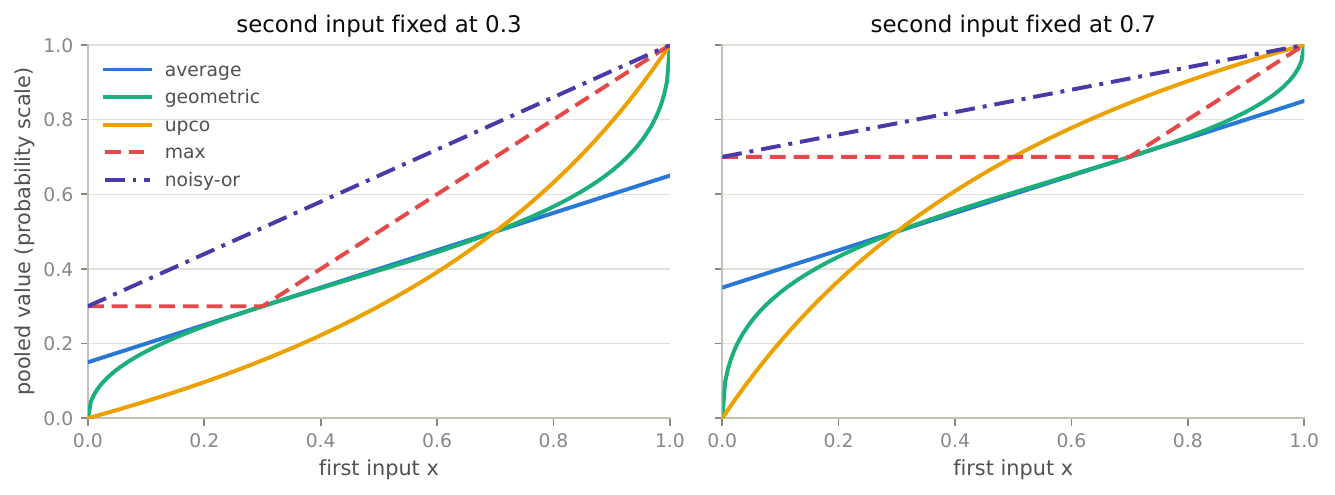}
\caption{The probability-scale pools on the same inputs: pooled value against the first input
$x$, with the second input fixed at 0.3 (left) and 0.7 (right). Dash patterns distinguish the
maximum and the probabilistic sum, which coincide with other rules on segments.}
\label{fig:rulesprob}
\end{figure}

\begin{figure}[H]
\centering
\includegraphics[width=\textwidth]{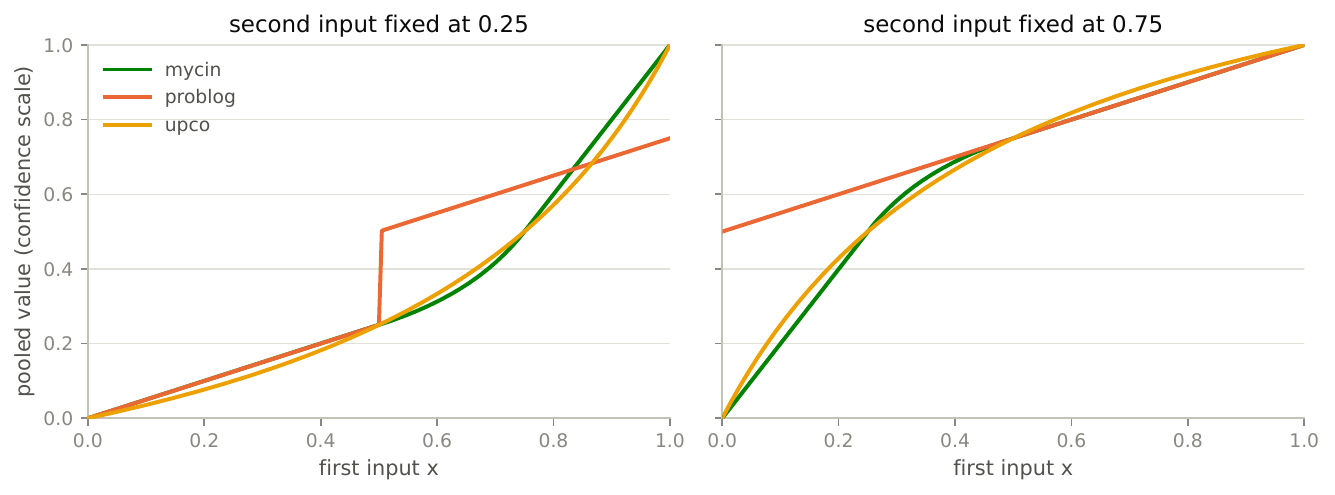}
\caption{The confidence-scale combinators (0.5 = no information) with the second input fixed
at 0.25 (left; the fixed input opposes) and 0.75 (right; the fixed input supports), with upco
for comparison. The MYCIN rule is continuous across $x=1/2$; the ProbLog rule jumps there in
the left panel, where its conflict branch takes over.}
\label{fig:rulesconf}
\end{figure}

\subsection{A decision-equivalence proposition}
\label{sec:prop}

The proposition below compares the formulas on the same two numbers. In their usual
applications, the first three combine probabilities and MYCIN combines signed evidence
strengths (Section~\ref{sec:types}).

\begin{proposition}\label{prop:half}
For all $p_1,p_2\in(0,1)$: $\mathrm{avg}(p_1,p_2)>1/2$, $\mathrm{geo}(p_1,p_2)>1/2$,
$\mathrm{upco}(p_1,p_2)>1/2$, and the signed strength output of the MYCIN combination is
positive (equivalently, its confidence output exceeds $1/2$), each exactly when $p_1+p_2>1$.
The same holds with $>$ replaced by $<$ or $=$.
\end{proposition}

\begin{proof}
For the average the claim is immediate. For upco, $p_1p_2>(1-p_1)(1-p_2)$ simplifies to
$p_1+p_2>1$. For the geometric pool the same inequality appears under square roots. For MYCIN:
if both strengths are positive or both negative the output has that sign, and both reports are
on the same side of $1/2$, so $p_1+p_2>1$ holds exactly for two positives; in the conflict case
with $p_1\ge 1/2\ge p_2$ the output sign is the sign of $(2p_1-1)-(1-2p_2)$, which is the sign
of $p_1+p_2-1$.
\end{proof}

All four formulas yield the same binary decisions at threshold $1/2$ on the same
numerical inputs. Their numerical outputs differ and can be distinguished by probability-sensitive
scores or other thresholds. These numerical comparisons show algebraic behaviour only; they do
not imply that the rules normally take the same kind of input.

\FloatBarrier
\section{Experiments I: matching rules to generating mechanisms}
\label{sec:corr}

Two terms are used throughout this section and the next. A \emph{world} is a complete
probability model for the two reports and the outcome. A rule is \emph{exact} in a world
when it returns the conditional probability of the outcome from the information supplied
to the rule. Each experiment derives this probability analytically and compares it with
the simulated frequency and the candidate formulas on the same report values. In
Section~\ref{sec:betting}, a heading such as ``exact rule: the average'' means that the
average is the conditional probability in that setting; individual outcomes remain
random.

\subsection{Average: two possible readings of an ambiguous query}
\label{sec:avgcorr}

Let $Z\in\{1,2\}$ be the intended reading of an ambiguous query, with
$P(Z=1)=P(Z=2)=1/2$; $Z$ is not observed. The two numbers are conditional reports, one per
reading: the claim is true with probability $0.30$ under reading 1 and $0.80$ under
reading 2 (a question about ``Cambridge'' may concern Cambridge, England or Cambridge,
Massachusetts). Each report is calibrated conditional on its reading, and neither source
knows the round's actual reading. The probability that the claim is true is $0.5\cdot 0.30+0.5\cdot 0.80=0.55$, the
average.

The situation generalizes: when the two reports describe two mutually exclusive circumstances
and circumstance $i$ holds with probability $w_i$, the correct combination is the weighted
average, the weights being the probabilities of the reporters' circumstances. The same
calculation applies to RAG marginalization when the retrieval weights are posterior
probabilities of the relevant passage and the answer probabilities are conditional on
that passage \citep{lewis2020}; plain averaging of retrieved passages need not satisfy these
conditions. Other examples are gated mixtures of experts \citep{jacobs1991}, Bayesian
model averaging \citep{hoeting1999,wilsonizmailov2020}, and ensemble averaging with the
member index as the unobserved alternative \citep{lakshminarayanan2017,ovadia2019}.
In the simulation, observed frequencies match the average for all five report pairs; the
other candidate formulas differ from the observed frequency (the probabilistic sum, for
example, gives 0.86 where the exact value is 0.55).

When the reports are independent evidence about the same case, the average underestimates
the posterior probability (Section~\ref{sec:upcocorr}).

\subsection*{Method notes}

The reported frequencies use $10^5$ to $4\cdot 10^5$ rounds per configuration, giving
binomial standard errors below $0.0016$; the observed frequencies match the stated
formulas within a few standard errors. The formulas are derived analytically; the
simulations check their implementation and measure the errors of mismatched rules. The
derived conditional probability is the reference used in Section~\ref{sec:betting}. Each
mechanism defines the inputs independently of the formula under study and derives the
outcome probability from the source relation. The code of each experiment is one short
file in the public repository (Appendix~A).

\subsection{Maximum and probabilistic sum: two attempts at the same task}
\label{sec:trapscorr}

Consider the event that at least one of two generated solutions passes the task's test
suite. The system samples two candidate programs for the same
task; the two numbers are the \emph{per-attempt} success probabilities, 0.30 and 0.80 in the
example. Pass@$k$ evaluation of code generation reports the probability that at least one attempt
passes \citep{chen2021}. The two per-attempt success probabilities are the inputs here;
their dependence determines the union probability.

The nested and independent cases below are two reference cases. Positive dependence
generally gives a value between the maximum and the probabilistic sum; negative
dependence can give a larger union probability.
If the attempts come from two unrelated models whose failures are independent, the event fails
only if both attempts fail:
$1-0.70\cdot 0.20=0.86$, the probabilistic sum (the independent-attempts pass@2).

In the nested model, each task has a sampled difficulty, and attempt $i$ succeeds when
the difficulty is below $p_i$. Every task solved by the weaker attempt is also solved by
the stronger one, so ``at least one passes'' is the same event as ``the stronger attempt
passes'', with probability $\max(0.30,0.80)=0.80$. In the simulation the coupled case
uses one shared uniform variable.

The upper Fr\'echet end $\min(1,p_1+p_2)$ of Section~\ref{sec:tconorms} is reached
under a maximally negatively dependent construction (the simulation uses one shared
uniform variable with opposite thresholds). Applications include the leaky noisy-or that aggregates per-finding
probabilities into a case-level probability \citep{liao2019}, and pass@$k$-style reasoning,
where additional samples help less when the attempts are correlated
(Section~\ref{sec:cumulcorr}).

\subsection{The ProbLog conflict rule: separate supporting and opposing searches}
\label{sec:problogcorr}

On the confidence scale, a source with confidence $c$ contributes evidence of strength
$s=2c-1$; the strength is read as the probability that the corresponding search returns
evidence with that sign. Two independent searches examine a generated claim: one returns a
supporting passage with probability $a$, the other a contradicting passage with
probability $b$. This models fact verification with separate
searches for supporting and contradicting evidence \citep{thorne2018}.

The agreement case (two supporting searches) counts the rounds where at least one search
returns a passage: $a+b-ab$, the probabilistic sum. The conflict case counts the rounds
where the supporting search returns a passage and the contradicting search does not:
$a(1-b)$, which is the ProbLog value. Both
counts match the simulation exactly. The correspondence presumes independence of the two
searches.

\subsection{Multiplicative pooling: two reports based on independent evidence}
\label{sec:upcocorr}

Each round has $P(Y=1)=1/2$. Model 1 observes the image attached to the case and model 2
the accompanying text; the two observations are conditionally independent given the truth
value, under both values of $Y$, and both models compute their posteriors from the same
prior $1/2$.
Each model receives a signal of known reliability (drawn between 0.55 and 0.95) and reports
the calibrated posterior probability for that signal.

Under these conditions the posterior probability given both observations is exactly
$\mathrm{upco}(p_1,p_2)$: on the odds scale, each report contributes its likelihood ratio, and
likelihood ratios of independent observations multiply. The simulated frequencies match
the posterior for every signal pattern. For $p_1=p_2=0.8$ the exact posterior is $0.94$; the average is
$0.80$. This is the product of
experts \citep{hinton2002} in its two-source form, and the naive-Bayes fusion of two
conditionally independent modalities.

With prior $0.3$ in the simulation, each calibrated report already contains the same
prior odds; multiplying the reports counts them twice, and upco is systematically too
high. The odds pool with the prior correction of Section~\ref{sec:upco} restores the
exact posterior. A shared base rate is one example of a common prior; the correction
counts it once per combination. It removes duplicate prior odds only; dependence caused
by shared data, architecture, or correlated errors requires a separate model.

\subsection{Geometric pooling: doubled log-odds and averaged samples}
\label{sec:geocorr}

We consider two exact cases: uniformly doubled binary log-odds, and posterior
distributions from averaged samples.

First, for single binary reports, geometric pooling is exact when both sources report log-odds
twice as large as their evidence justifies. In the two-models situation of
Section~\ref{sec:upcocorr}, let each model report double its justified log-odds (a model
whose observation justifies odds 4:1 reports odds 16:1). This
is a form of overconfidence addressed by temperature scaling, which divides the network's
log-odds by a fitted constant \citep{guo2017}. Averaging the two
reports on the log-odds scale halves each doubled logit and adds them, recovering exactly
the calibrated combined answer. Geometric pooling
is therefore the exact rule for sources that double their log-odds; averaging member
networks' logits, a standard ensemble variant, is equivalent to geometric pooling for
binary predictions.\footnote{For multiclass softmax outputs the analogue is the
normalized geometric mean of the class probabilities.} Section~\ref{sec:knob} develops
the general version.

Second, for posterior distributions, geometric pooling gives the posterior
corresponding to the average of two equal-sized samples. Following
\citet{pettigrewweisberggeo}: a model fails on tasks of a given kind at an unknown rate;
two evaluation runs each observe a sample (12 failures out of 20 tasks, and 18 out of 20
in the second run) and yield posterior distributions over the failure rate from a shared
flat prior. The geometric pool of the two distributions equals the posterior
distribution of the averaged sample,
15 of 20. Multiplying the distributions gives the posterior of the added samples (30 of 40), exact
only when the two task sets are disjoint \citep{scott2016,neiswanger2014}; under unknown
overlap the averaged-sample form avoids double counting \citep{julieruhlmann1997}. The identity is algebraic and
holds for every data pair. The linear
pool of the two distributions is a mixture of them and equals the averaged-sample
posterior only when the two samples coincide.

\subsection{The MYCIN rule: Poisson evidence and item-level refutation}
\label{sec:mycincorr}

The agreement case of the MYCIN rule is the probabilistic sum and matches the
independent-signs count of Section~\ref{sec:problogcorr}. The conflict case has an exact
reformulation. Define $\lambda(x)=-\ln(1-x)$, so that
$x=1-e^{-\lambda}$ is the probability that at least one event of a Poisson stream of rate
$\lambda$ occurs in a unit period. Then
\begin{equation}
\frac{a-b}{1-b}\;=\;1-e^{-(\lambda_a-\lambda_b)},
\label{eq:netrate}
\end{equation}
which is verified by substituting $e^{-\lambda_a}=1-a$ and $e^{-\lambda_b}=1-b$: the
MYCIN conflict value is the probabilistic-sum value of the net rate
$\lambda_a-\lambda_b$. Each evidence strength maps to a Poisson rate, same-sign
combination adds the rates, and the conflict formula subtracts the refutation rate; here
$\lambda$ is the expected number of evidence arrivals in one unit interval.

Independent support and opposition instead give $a(1-b)$, the ProbLog value. The
following coupled process has the MYCIN conflict value as its outcome probability.

\begin{proposition}\label{prop:refutation}
Let $0<b<a<1$ and set $\lambda_a=-\ln(1-a)$, $\lambda_b=-\ln(1-b)$. Let supporting items
arrive as a Poisson process of rate $\lambda_a$, and let each supporting item, independently
of everything else, be refuted with probability $\lambda_b/\lambda_a$. Then (i) the probability
that at least one supporting item arrives is $a$; (ii) the refuted items form a Poisson
process of rate $\lambda_b$, so the probability that at least one refutation occurs is $b$;
and (iii) the probability that at least one supporting item survives unrefuted is
\begin{equation*}
P(\text{at least one supporting item survives})
=1-e^{-(\lambda_a-\lambda_b)}=\frac{a-b}{1-b}\,,
\end{equation*}
the MYCIN conflict value.
\end{proposition}

\begin{proof}
(i) is the definition of $\lambda_a$. By the thinning property of Poisson processes
\citep[ch.~5]{kingman1993}, the refuted and the surviving items form independent Poisson
processes with rates $\lambda_a\cdot(\lambda_b/\lambda_a)=\lambda_b$ and
$\lambda_a-\lambda_b$; this gives (ii), and (iii) follows from
$1-e^{-(\lambda_a-\lambda_b)}=1-(1-a)/(1-b)=(a-b)/(1-b)$.
\end{proof}

Under this construction the MYCIN conflict value is the probability that at least one
supporting item survives refutation, with both marginal reports calibrated.
The construction can represent citation checking in which each criticism targets a
specific supporting citation \citep{irving2018} and the answer is accepted when at least
one citation remains unrefuted. The construction assumes that every refutation targets a supporting item;
attacks do not arrive independently. Independent counter-evidence gives the ProbLog world of
Section~\ref{sec:problogcorr} instead. The simulated survival frequency matches the
formula, and both marginal reports are calibrated as stated.

Independent opposition yields the ProbLog value; item-level refutation yields the MYCIN
value. The construction models survival of the
supporting evidence; a truth probability for the claim itself requires a separate model,
and for calibrated models reporting on the fact the applicable combination remains the
odds pool of Section~\ref{sec:upcocorr}.

\subsection{Cumulation: two reasoning chains with unobserved overlap}
\label{sec:cumulcorr}

Two reasoning chains support the same answer, and they share evidence: chain 1 uses a
retrieved passage and a code check, chain 2 uses the same code check and a knowledge-base
lookup. Each item of evidence is valid with a known probability; a chain holds if all its
items are valid; the answer holds if at least one chain does. If the chains shared no items,
the exact combination of the two chain confidences would be the probabilistic sum; if they
were copies of each other, the maximum. Under partial sharing the union probability lies
between these endpoints, and the parameter $a$ of cumulation selects a point on this
interval.

We evaluate two sharing models. In the first, sharing is all or nothing: with probability
$1-a$ both chains use the same uniform variable, and with probability $a$ they are
independent, as resampled generations are when they reuse the same retrieved passage or
fall back on the same reasoning plan. Here the probability that at least one chain
holds is exactly $\mathrm{cumul}(c_1,c_2,a)$, for every $a$, because the formula is
linear in $a$. In the second, sharing is item-level: each chain has $m=4$ evidence items
and the chains hold $k$ of them in common. The endpoints again match ($k=0$ gives the
probabilistic sum, $k=4$ the maximum), but between them cumulation with the parameter
$a=1-k/(2m-k)$ proposed by \citet{tammet2021} overestimates the simulated probability by
up to 0.02. The rule is exact for all-or-nothing sharing and an approximation for
item-level sharing. The same distinction matters in self-consistency decoding
\citep{wang2023}: correlation reduces the gain from additional chains. Section~\ref{sec:gkcorr} computes the overlap from explicit shared-event
identities. When event identities are unavailable, cumulation provides an
approximation.

\subsection{The GK calculation: two chains with identified shared premises}
\label{sec:gkcorr}

Use the item-sharing model of Section~\ref{sec:cumulcorr}, with the shared item
identities observed (chain 1 uses a retrieved passage and a code check, chain 2 the same
code check and a knowledge-base lookup). Treat each evidence item as an uncertain ground
premise (a GK activation event) and each chain as a retained proof. The shared part has
probability $s$, the product of the shared
items' validities, and the probability that at least one chain holds is
$\mathrm{gk}(c_1,c_2;s)$ of Section~\ref{sec:gk}, with the shared-event dependence computed from the item identities.

Three simulations compare the GK formula with cumulation.

\paragraph{Equal item probabilities.}
In the item-level grid of Section~\ref{sec:cumulcorr} ($m=4$ items per chain, each valid
with probability $0.85$, $k$ shared), the GK value matches the simulated frequency for
\emph{every} $k$ (largest absolute deviation $0.0007$, within the Monte Carlo sampling
error), where cumulation with the count-based parameter has an error of up to $0.021$.

\paragraph{Unequal item probabilities.}
For chains $\{0.95\}\cup\{0.9,0.75\}$ and $\{0.8,0.85\}\cup\{0.9,0.75\}$ the simulated
frequency is $0.6642$, the GK value $0.6642$, and cumulation with the count parameter
$0.7400$, an overestimate of $0.0758$. The count parameter uses only the number of shared
items. The product $s$ also uses their probabilities; in this case the shared items have
high probabilities.

\paragraph{Three proofs.}
Three chains share in a chain pattern: $A$ shares one item with $B$, and $B$ another with
$C$; $A$ and $C$ are disjoint (every item at $0.85$, seven distinct items in all). The
exact union probability is $0.8984$. Inclusion--exclusion over the three availability
events, with each intersection evaluated from the union of the corresponding item sets,
returns this value. Folding the two-argument formula through the pairwise overlaps gives
$0.8318$, and ignoring the second overlap gives $0.9169$. Both procedures discard the
event identities needed for the third proof.

In self-consistency decoding \citep{wang2023}, the same calculation applies when sampled
chains reuse identified passages or subresults and per-item validity probabilities are
available. The premise identities determine the overlap; the premise probabilities still
require calibration. These simulations and world 10 of Section~\ref{sec:worlds} evaluate
the GK formula independently of GK's proof search and opposition handling.

\subsection{Dempster--Shafer: opposing sources conditioned on non-conflict}
\label{sec:dscorr}

Model 1 supports the claim and model 2 opposes it. Let $G_i$ indicate that model $i$ has
decisive evidence, with $P(G_i=1)=m_i$ independently; a model with decisive evidence is
correct, and otherwise it is uninformative. Since the two statements contradict, both
models cannot have decisive evidence, and the construction conditions on $G_1G_2=0$:
draws with both decisive are discarded and redrawn. The redraw is rejection sampling that implements
Dempster's normalization: conditioning on non-conflict.

Conditioning on non-conflict leaves three cases: only model 1 decisive (the claim is
true), only model 2 decisive (the claim is false), and neither decisive. The frequency of
the first case is $\mathrm{bel}$; adding the neither-decisive case gives $\mathrm{pl}$;
resolving the neither-decisive case by a fair coin gives $\mathrm{BetP}$. All three
simulated frequencies agree with the formulas within Monte Carlo sampling error, for
example at $(m_1,m_2)=(0.8,\,0.3)$: observed $0.7359\,/\,0.9205\,/\,0.8286$
against $\mathrm{bel}=0.7368$, $\mathrm{pl}=0.9211$, $\mathrm{BetP}=0.8289$. The agreement case
(both models supporting) contains no contradiction, nothing is discarded, and the combined
support is the probabilistic sum.

The construction also illustrates Zadeh's (\citeyear{zadeh1984}) criticism: when $m_1$
and $m_2$ are both large, little mass remains after conditioning, and few simulated
draws are retained. This model applies only when both sources cannot
have decisive evidence at the same time.

\subsection{Summary of part one}
\label{sec:corrsummary}

The world numbers refer to the betting worlds of Section~\ref{sec:worlds}, which reuse
these mechanisms.

\begin{center}
\small
\setlength{\tabcolsep}{4pt}
\begin{tabular}{l>{\raggedright\arraybackslash}p{9cm}c}
\toprule
rule & setting in which the rule is exact & world\\
\midrule
average & two readings, one meant, equally likely & 1\\
weighted average & as above, unequal; weights $=$ reading probabilities & --\\
maximum & two attempts, one success event contained in the other & 5\\
probabilistic sum & at least one of two independent attempts & 4\\
ProbLog conflict & independent supporting and opposing searches & 7\\
upco (multiplicative) & two posterior reports from conditionally independent evidence, prior $1/2$ & 2\\
odds pool & as above, any prior & 2$'$\\
geometric & doubled log-odds reports; averaged samples, shared prior & 3\\
MYCIN & agreement: independent evidence; conflict: each opposing item refutes one supporting item & 6\\
cumulation & mixture of complete overlap and independence; approximate for fixed shared items & 8\\
GK calculation & identified shared premises; all proofs combined together beyond two & 10\\
DS bel, pl, BetP & opposing sources conditioned on not both having decisive evidence; remaining ignorance divided equally & 9\\
\bottomrule
\end{tabular}
\end{center}

Figure~\ref{fig:calibration} gives reliability diagrams for the two-models and
alternative-reading mechanisms. Each rule's pooled values are grouped into 25 equal-count
bins and the observed frequency of the event in each bin is plotted against the bin's
mean pooled value. This is a reliability diagram \citep{guo2017}: if a rule returns the
correct conditional probability, its reliability curve lies on the diagonal. In the two-models world (left) upco lies on the diagonal, the
average and the geometric pool are underconfident for pooled values above $1/2$ and
overconfident below it, and the probabilistic sum shows the largest deviation. The
largest MYCIN bin deviation is $0.026$, against $0.011$ for upco: close to the diagonal
without matching it. In the alternative-reading world (right) the average and geometric
curves (nearly indistinguishable there) lie near the diagonal, while upco and MYCIN
deviate from it. Calibration depends on how the reports and outcomes are generated: a
rule calibrated under one dependence structure may be miscalibrated under another.
Calibration results are specific to the benchmark's distribution of reports and
dependence between sources.

\begin{figure}[H]
\centering
\includegraphics[width=\textwidth]{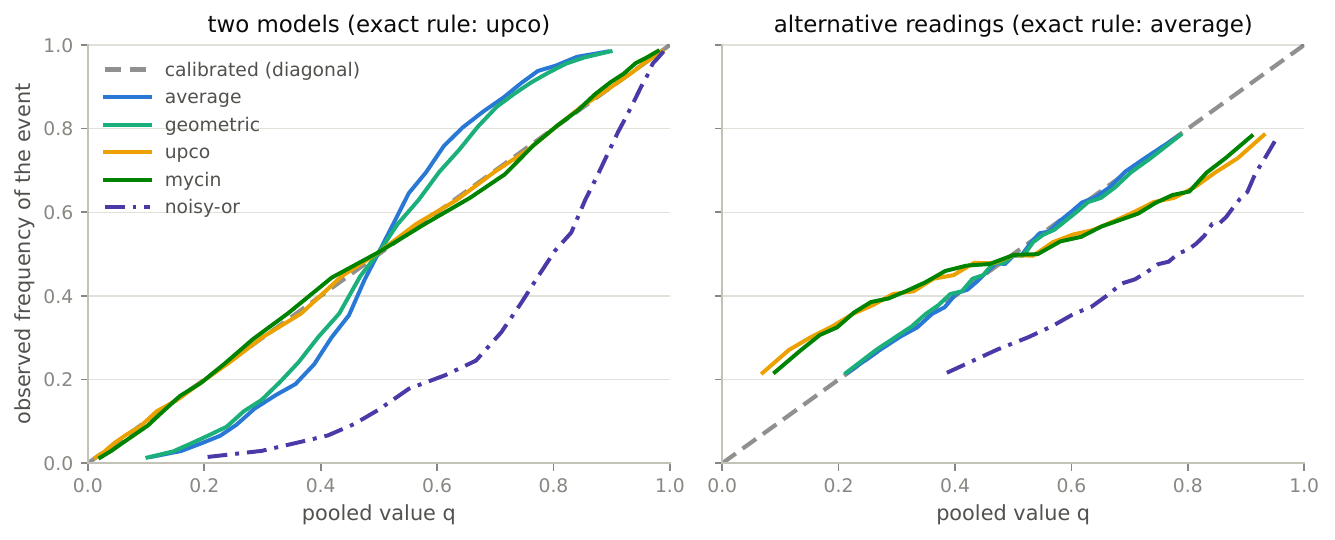}
\caption{Observed frequency of the event against the pooled value (25 equal-count bins
per rule) in the two-models world (left) and the alternative-reading world with uniformly
drawn reports (right). A point $(q,r)$ on a curve means that the cases assigned pooled
values near $q$ produced the event with frequency $r$; calibration is $r=q$. The rule
derived for each setting lies on the dashed diagonal: upco is calibrated in the
two-models setting, and the average in the alternative-reading setting.}
\label{fig:calibration}
\end{figure}

\FloatBarrier
\section{Experiments II: the cost of a mismatched rule}
\label{sec:betting}

\subsection{Logarithmic score and its betting interpretation}
\label{sec:game}

Section~3 established which pooling rule recovers the correct conditional
probability in each generating mechanism.  That distinction matters only to
the extent that different reported probabilities have consequences.  In
particular, several rules make the same decision at threshold $1/2$, so
ordinary binary accuracy at that threshold cannot measure the cost of using
the wrong rule.  We therefore reuse the same worlds and evaluate the reported
probabilities themselves, and the decisions they induce at different cost
thresholds.

The choice of evaluation criterion does not determine which rule is correct
in a generating mechanism.  Once the conditional probability in a world is
$p$, reporting $p$ is optimal under any strictly proper scoring rule; different
proper scores differ in how they weight deviations from $p$.  We use the
logarithmic score because it has a direct interpretation as long-run wealth
growth under Kelly betting.  Section~4.3 gives a complementary
decision-theoretic evaluation at varying cost thresholds, and Section~4.4
relates these threshold losses to proper scoring rules.

We first use logarithmic score, with Kelly wealth growth as a monetary
interpretation.

For the logarithmic score, consider the following betting realization.  Each
round, before the outcome is known, two sources provide their reports and a
pooling rule combines them into $q$.  At even odds, the current wealth is split
between the two outcomes: fraction $q$ is allocated to the event and fraction
$1-q$ to its complement.  If the event occurs, wealth is multiplied by $2q$;
otherwise it is multiplied by $2(1-q)$.  For example, with current wealth 100
and $q=0.8$, the two stakes are 80 and 20; wealth becomes 160 if the event
occurs and 40 otherwise.

This allocation is chosen because its logarithmic wealth increment is exactly
the binary logarithmic score, up to an additive constant; consequently its
expected value is uniquely maximized when $q$ equals the true conditional
probability.

Log wealth is measured in bits: a change of $+1$ bit doubles wealth, a change of $-1$ bit
halves it, and log changes add across rounds --- doubling and then halving gives $1-1=0$
bits and returns wealth to its starting value. The example's multipliers are
$\log_2 1.6=0.678$ bits and $\log_2 0.4=-1.322$ bits, and the reported score is the average
of such per-round log multipliers. The total log growth over $T$ rounds and its per-round
mean are
\begin{equation*}
B_T=\log_2\frac{W_T}{W_0}=\sum_{t=1}^{T}\log_2\frac{W_t}{W_{t-1}},
\qquad
\bar b_T=\frac{B_T}{T}:
\end{equation*}
$B_T=1$ means that wealth doubled over the run, and a mean of $g$ bits per round for $T$
rounds means $B_T=gT$ and $W_T/W_0=2^{gT}$; $0.01$ bits per round sustained for 100 rounds
is one bit in total, a factor of two. Positive mean log growth indicates geometric growth; negative mean log growth indicates
geometric decline. Absolute mean log
growth depends on the predictive information in the world; the oracle gap defined below
isolates the pooling error.

In per-round terms, let $y\in\{0,1\}$ be the outcome and $q$ the pooled probability.
At even odds the per-round growth, in bits, is
\begin{equation*}
g(q,y)=y\log_2(2q)+(1-y)\log_2\bigl(2(1-q)\bigr),
\end{equation*}
the base-2 logarithm of the wealth multiplier; $q=1/2$ gives zero for either outcome. The
score of a rule is the average of $g$ over the rounds. When the true probability of the
event is $p$, the expected score of quoting $q$ is
\begin{equation*}
G(p,q)=p\log_2(2q)+(1-p)\log_2\bigl(2(1-q)\bigr).
\end{equation*}
Apart from the constant introduced by the betting odds, $-g$ is the binary log loss, so
differences in $G$ are differences in expected log loss. The expected score is maximized
exactly at $q=p$, and the shortfall is the binary Kullback--Leibler divergence in bits,
\begin{equation*}
G(p,p)-G(p,q)=D_2(p\,\|\,q)
=p\log_2\frac{p}{q}+(1-p)\log_2\frac{1-p}{1-q}\;\ge\;0,
\end{equation*}
with equality exactly when $q=p$. The oracle gap is nonnegative; it equals both the loss
in expected log-wealth growth and the excess binary log loss.

This is a Kelly-style betting realization of the logarithmic score.
Under the two-outcome allocation above, quoting the true probability
maximizes expected logarithmic wealth growth; Kelly's original analysis
establishes the more general connection between probabilistic information,
betting odds, and maximal long-run growth \citep{kelly1956}.  Interior
forecasts $0<q<1$ retain positive wealth after either outcome, whereas a
forecast of $0$ or $1$ loses all wealth when wrong.  The experimental
protocol therefore clips every pooled value to
$[10^{-6},\,1-10^{-6}]$.

Two reference strategies recur in every table. The \emph{probability oracle} uses the
true conditional probability $p$ of the round; it knows $p$, not the realized outcome.
The \emph{baseline} uses $q=1/2$, which keeps wealth constant at even odds and scores
zero.

The divergence identity extends to odds that are fair for any base rate
\citep[ch.~6]{coverthomas2006}: changing the bookmaker's odds adds the same
outcome-dependent term to every rule's growth, so the oracle gap is invariant to the
odds. World $2'$ below therefore uses odds corresponding to base rate $0.3$.

The gap is also the excess binary cross-entropy on the log-wealth scale; the
growth--log-score equivalence is standard in forecasting
\citep{roulstonsmith2002,hagedornsmith2009,johnstone2012}, and binary cross-entropy is a
standard training loss for probabilistic classifiers \citep{lakshminarayanan2017}.

Proposition~\ref{prop:half} implies that its four formulas differ here only through
probability magnitude, since their even-odds decisions coincide.

\subsection{Betting results under ten generating mechanisms}
\label{sec:worlds}

For each of the main rules of Section~\ref{sec:corr} we run the betting game in a
setting where the rule is exact; the world numbers are those of the summary table in
Section~\ref{sec:corrsummary}. Every rule receives the same two reports; the oracle uses
the correct conditional probability, and the baseline uses $q=1/2$. Both table rows are
in bits per round, from runs of $10^5$ rounds.

How to read the tables:

\begin{center}
\small
\begin{tabular}{ll}
\toprule
item & meaning\\
\midrule
exact rule (heading) & the rule equal to the conditional probability in that setting\\
oracle/\emph{rule} (column head) & that rule coincides with the probability oracle\\
mean log growth & absolute score in this world, bits per round\\
oracle gap & excess log loss relative to the oracle, bits per round\\
baseline & $q=1/2$, zero growth at even odds\\
column order & decreasing realized mean log growth\\
\bottomrule
\end{tabular}
\end{center}

The oracle gap measures the excess log loss caused by the pooling rule; absolute mean
log growth also depends on how predictable the setting is. A gap of $d$ bits
per round compounds to a relative wealth factor of $2^{-dT}$ over $T$ rounds: for the
world-2 average, at $d=0.089$, the wealth ratio to the oracle is $2^{-0.89}\approx 0.54$
after ten rounds and about $1/478$ after 100.

\medskip
\noindent\textbf{World 1}, alternative readings of the question
(Section~\ref{sec:avgcorr}); exact rule: the average.
\begin{center}
\small
\begin{tabular}{lcccccc}
\toprule
 & oracle/average & geom. & baseline & upco & MYCIN & noisy-or\\
mean log growth & $+0.060$ & $+0.058$ & $0$ & $-0.025$ & $-0.032$ & $-0.250$\\
oracle gap & $0$ & $0.002$ & $0.060$ & $0.085$ & $0.092$ & $0.310$\\
\bottomrule
\end{tabular}
\end{center}
The average attains the oracle; at the reported $0.060$ bits per round, 100 rounds
correspond to $6$ bits of log growth, a wealth multiple of about $2^6=64$. Upco, MYCIN
and the probabilistic sum have negative mean log growth; the probabilistic sum has the
largest oracle gap.

\medskip
\noindent\textbf{World 2}, two separate models (Section~\ref{sec:upcocorr}); exact rule:
multiplicative (upco).
\begin{center}
\small
\begin{tabular}{lcccccc}
\toprule
 & oracle/upco & MYCIN & geom. & average & noisy-or & baseline\\
mean log growth & $+0.415$ & $+0.412$ & $+0.345$ & $+0.326$ & $+0.139$ & $0$\\
oracle gap & $0$ & $0.003$ & $0.070$ & $0.089$ & $0.276$ & $0.415$\\
\bottomrule
\end{tabular}
\end{center}
The average has an oracle gap of $0.089$ bits per round, 21\% of the oracle's mean log
growth; the MYCIN gap is $0.003$ bits per round.

Figure~\ref{fig:wealth} shows one simulated cumulative log-wealth path for each rule
over 500 rounds, in the two-models world and in the alternative-reading
world; the slopes estimate the mean log growth of the tables. In the two-models world all
displayed curves have positive slopes; upco has the largest slope and overlaps the
oracle's dashed curve. In the
alternative-reading world the average and the geometric pool increase, upco remains near zero
cumulative log growth, and the probabilistic sum decreases by about 90 bits over the same
rounds in which the oracle's cumulative log wealth increases by about 40 bits.

\begin{figure}[H]
\centering
\includegraphics[width=\textwidth]{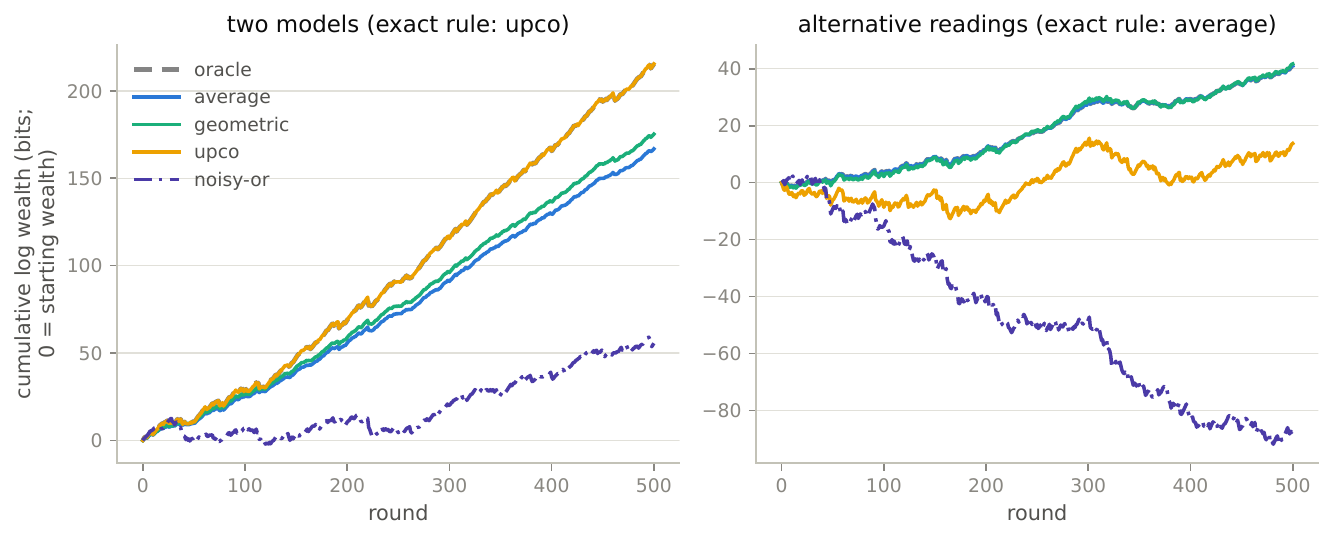}
\caption{One cumulative log-wealth path per pooling rule over 500 rounds at even odds
(bits; 0 = starting wealth), in the two-models world (left) and
the alternative-reading world (right). A vertical difference of one bit corresponds to a factor
of two in wealth; a
height of 10 is a wealth factor of $2^{10}=1024$.
Slopes equal the mean log growth reported in the world tables.
In the left panel the oracle's dashed curve is hidden beneath upco.}
\label{fig:wealth}
\end{figure}

\textbf{World 2$'$} uses the same source model with prior $0.3$ (odds priced
accordingly): the
odds pool attains the oracle growth $+0.358$; uncorrected upco reaches $+0.302$, an oracle gap
of $0.056$ bits per round from counting the prior twice.

\medskip
\noindent\textbf{World 3}, doubled log-odds reports (Section~\ref{sec:geocorr}); exact rule: the
geometric pool.
\begin{center}
\small
\begin{tabular}{lcccccc}
\toprule
 & oracle/geom. & average & MYCIN & upco & noisy-or & baseline\\
mean log growth & $+0.382$ & $+0.367$ & $+0.271$ & $+0.264$ & $+0.001$ & $0$\\
oracle gap & $0$ & $0.015$ & $0.111$ & $0.118$ & $0.381$ & $0.382$\\
\bottomrule
\end{tabular}
\end{center}

\medskip
\noindent\textbf{World 4}, independent attempts (Section~\ref{sec:trapscorr}); exact rule: the
probabilistic sum.
\begin{center}
\small
\begin{tabular}{lcccccc}
\toprule
 & oracle/noisy-or & average & geom. & MYCIN & upco & baseline\\
mean log growth & $+0.277$ & $+0.063$ & $+0.063$ & $+0.015$ & $+0.008$ & $0$\\
oracle gap & $0$ & $0.214$ & $0.214$ & $0.262$ & $0.269$ & $0.277$\\
\bottomrule
\end{tabular}
\end{center}

\medskip
\noindent\textbf{World 5}, difficulty-nested attempts (Section~\ref{sec:trapscorr}); exact rule:
the maximum.
\begin{center}
\small
\begin{tabular}{lcccccc}
\toprule
 & oracle/max & average & geom. & noisy-or & MYCIN & upco\\
mean log growth & $+0.125$ & $+0.062$ & $+0.062$ & $+0.038$ & $+0.014$ & $+0.007$\\
oracle gap & $0$ & $0.063$ & $0.063$ & $0.087$ & $0.111$ & $0.118$\\
\bottomrule
\end{tabular}
\end{center}

The geometric pool is exact for doubled log-odds, the probabilistic sum for independent
attempts, and the maximum for nested attempts.

\medskip
\noindent\textbf{World 6}, review with one-to-one refutation (Section~\ref{sec:mycincorr});
exact rule: the MYCIN conflict rule. Inputs are the two evidence strengths; the table
compares the conflict rules.
\begin{center}
\small
\begin{tabular}{lcccc}
\toprule
 & oracle/MYCIN & ProbLog & baseline & naive\\
mean log growth & $+0.140$ & $+0.051$ & $0$ & $-0.029$\\
oracle gap & $0$ & $0.089$ & $0.140$ & $0.169$\\
\bottomrule
\end{tabular}
\end{center}

\medskip
\noindent\textbf{World 7}, review with independent counter-evidence
(Section~\ref{sec:problogcorr}); exact rule: the ProbLog rule.
\begin{center}
\small
\begin{tabular}{lcccc}
\toprule
 & oracle/ProbLog & naive & baseline & MYCIN\\
mean log growth & $+0.074$ & $+0.023$ & $0$ & $-0.023$\\
oracle gap & $0$ & $0.051$ & $0.074$ & $0.097$\\
\bottomrule
\end{tabular}
\end{center}
Worlds 6 and 7 use identical inputs and differ only in how the opposing evidence relates
to the supporting evidence: MYCIN is exact under item-level refutation, the ProbLog rule
under independent opposition; the naive difference is exact in neither setting.

\medskip
\noindent\textbf{World 8}, two chains sharing evidence with a known probability $a$ of
independence (Section~\ref{sec:cumulcorr}); exact rule: cumulation.
\begin{center}
\small
\begin{tabular}{lccccc}
\toprule
 & oracle/cumul & max & noisy-or & average & baseline\\
mean log growth & $+0.326$ & $+0.309$ & $+0.302$ & $+0.162$ & $0$\\
oracle gap & $0$ & $0.017$ & $0.024$ & $0.164$ & $0.326$\\
\bottomrule
\end{tabular}
\end{center}
World 8 uses one fixed overlap parameter, the probability $a$ that the chains are
independent. World 10 below computes the overlap separately from the shared premises in
each round.

\medskip
\noindent\textbf{World 9}, contradicting models (Section~\ref{sec:dscorr}); exact rule: the
pignistic value BetP.
\begin{center}
\small
\begin{tabular}{lccccccc}
\toprule
 & oracle/BetP & bel & pl & MYCIN & baseline & ProbLog & naive\\
mean log growth & $+0.227$ & $+0.171$ & $+0.121$ & $+0.065$ & $0$ & $-0.068$ & $-0.338$\\
oracle gap & $0$ & $0.056$ & $0.106$ & $0.162$ & $0.227$ & $0.295$ & $0.565$\\
\bottomrule
\end{tabular}
\end{center}
Here $\mathrm{bel}<\mathrm{BetP}<\mathrm{pl}$; both bounds yield positive mean log
growth, while the rules without the renormalization have negative mean log growth. Under the pignistic decision convention
used here, $\mathrm{bel}$ and $\mathrm{pl}$ are bounds and $\mathrm{BetP}$ is the
decision probability (Section~\ref{sec:ds} lists competing decision criteria).

\medskip
\noindent\textbf{World 10}, two derivations sharing visible evidence
(Section~\ref{sec:gkcorr}); exact rule: the GK calculation. Each round builds two
four-item derivations. The number of shared items $k$ and the item validities are
resampled each round, with validities drawn from $U(0.6,0.95)$. The inputs are $c_1$,
$c_2$, and the shared-premise probability $s$, computed from the derivation event sets.
World 8 uses one fixed dependence parameter; here the overlap varies by round.
\begin{center}
\small
\setlength{\tabcolsep}{4pt}
\begin{tabular}{lccccccc}
\toprule
 & oracle/GK & cumul (count) & cumul (0.5) & max & baseline & noisy-or & average\\
mean log growth & $+0.048$ & $+0.047$ & $+0.031$ & $+0.005$ & $0$ & $-0.003$ & $-0.025$\\
oracle gap & $0$ & $0.001$ & $0.017$ & $0.043$ & $0.048$ & $0.051$ & $0.073$\\
\bottomrule
\end{tabular}
\end{center}
The GK value matches the oracle ($+0.0480$ bits per round against $+0.0466$ for cumulation
with the count-based parameter). With item probabilities drawn from this distribution,
count-based cumulation has an oracle gap of $0.0014$--$0.0018$ bits per round (five
random seeds); the asymmetric configurations of Section~\ref{sec:gkcorr} produce larger errors.
The probabilistic sum and the average have negative mean log growth while the oracle's is
positive.

\medskip
Across worlds 1--5 the five recurring formulas allow a direct cross-world comparison; rows
are worlds, columns rules, entries oracle gaps in bits per round (the exact rule of world
5, the maximum, is not in the column set):
\begin{center}
\small
\begin{tabular}{lccccc}
\toprule
world (exact rule) & average & geom. & upco & MYCIN & noisy-or\\
\midrule
1 (average) & 0 & 0.002 & 0.085 & 0.092 & 0.310\\
2 (upco) & 0.089 & 0.070 & 0 & 0.003 & 0.276\\
3 (geometric) & 0.015 & 0 & 0.118 & 0.111 & 0.381\\
4 (noisy-or) & 0.214 & 0.214 & 0.269 & 0.262 & 0\\
5 (maximum) & 0.063 & 0.063 & 0.118 & 0.111 & 0.087\\
\bottomrule
\end{tabular}
\end{center}

\subsection{Threshold decisions and cost-loss results}
\label{sec:costloss}

The second evaluation uses threshold decisions and their costs. A system receives the
two reports about whether an answer is faulty and decides whether to send the answer for
human review or to release it. Review costs $C$; releasing a faulty answer costs $L>C$.
We set $L=1$ and write $c=C/L$: the cost of review as a fraction of the loss it prevents.
The costs per round:
\begin{center}
\small
\begin{tabular}{lcc}
\toprule
 & no fault & fault\\
\midrule
send for review & $c$ & $c$\\
release the answer & $0$ & $1$\\
\bottomrule
\end{tabular}
\end{center}
If the pooled probability of a fault is $q$, the expected cost of review is $c$ and the
expected cost of releasing is $q$, so the answer is sent for review exactly when $q>c$:
with $q=0.30$, review is chosen when $c=0.20$ and not when $c=0.40$. This is the standard
cost-loss model \citep{murphy1977,richardson2000,wilks2001} and corresponds to selective
prediction with escalation \citep{geifman2017}.

Expenses are reported relative to perfect foresight, which incurs cost $c$ on exactly
the faulty rounds and zero otherwise. Additional expense is $c$ for an unnecessary review and
$1-c$ for releasing a faulty answer (the full loss 1 where foresight pays $c$). We score a rule by this
\emph{avoidable expense}: the average expense beyond perfect foresight, per round; lower
is better. Subtracting the foresight spending
changes no comparison (it is the same for every rule) and makes the scale start at zero.
The candidate rule uses the pooled $q$. The probability oracle uses the conditional
probability $p$. Perfect foresight uses the realized outcome $y$ and defines the zero
reference of the expense scale. The probability oracle still incurs expense because the
outcome is unknown. Its curve is the minimum attainable by a probability
forecast: the true probability minimizes the avoidable expense at every cost ratio
simultaneously. A mismatched rule incurs additional cost at exactly those $c$
where $q$ and $p$ imply different actions.

The review protocol was run in worlds 1 and 2 at cost ratios 0.2, 0.35, 0.5, 0.65, 0.8. In
world 2 upco attains the minimal expense at every ratio; at ratio 0.5 the four compared
formulas tie to four decimal places (Proposition~\ref{prop:half}); at ratios 0.2 and 0.8
the average and the geometric pool incur additional expense (for example 0.182 against
the oracle minimum 0.165 at ratio 0.2). In world 1 the average equals the oracle minimum
and the multiplicative pool has higher expense at $c=0.2$ and $c=0.8$. The cost-loss tables express the
same forecast errors as threshold-dependent operational costs (Section~\ref{sec:schervish}
gives the formal correspondence).

Figure~\ref{fig:murphy} plots the same avoidable expense for every review-cost ratio
at once, one curve per rule. This plot is called a Murphy diagram \citep{ehm2016}.

Reading the figure. At $c=0.2$, each curve gives the average avoidable expense of its rule
when review costs 20\% of the prevented loss; the vertical gap from the oracle
curve is the extra expense caused by using that rule, incurred where $q$ and $p$ imply
different actions at that threshold. Repeating the
comparison for all $c$ gives the full curve. The shared curvature results from
subtracting the perfect-foresight cost and does not affect comparisons among the rules.
At $c=1/2$ the four formulas of Proposition~\ref{prop:half} induce the same decision and
their curves meet; the probabilistic-sum curve does not pass through that intersection.

\begin{figure}[H]
\centering
\includegraphics[width=\textwidth]{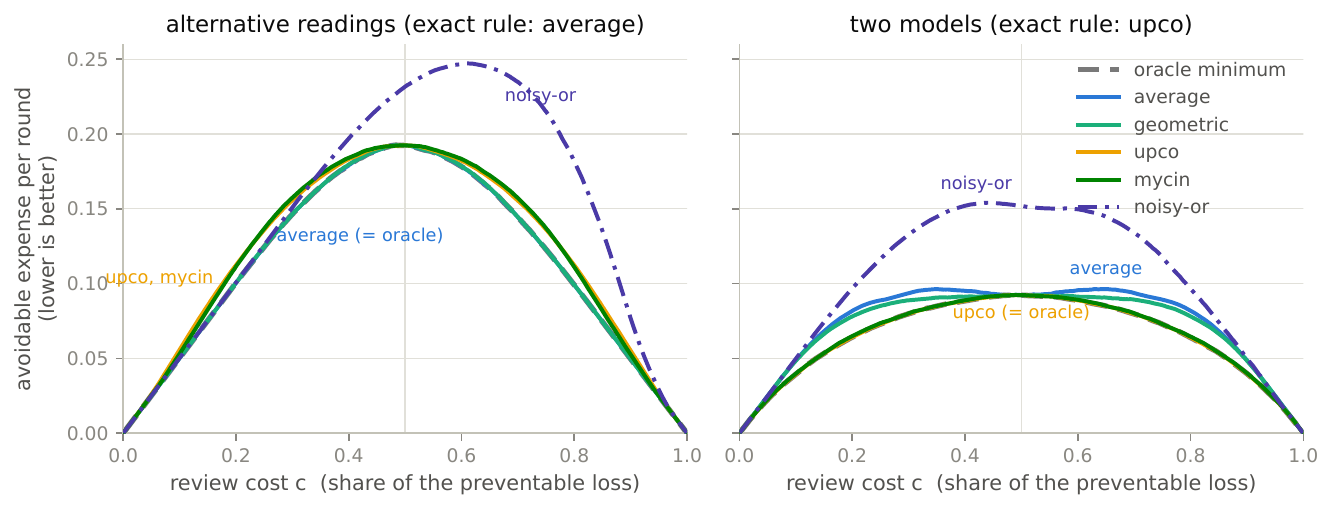}
\caption{Murphy diagrams for the alternative-reading world (left; exact rule: the average) and
the two-models world (right; exact rule: upco). Horizontal axis: cost ratio $c$, the cost
of review as a share of the prevented loss. Vertical axis: average avoidable
expense beyond perfect foresight; lower is better. The dashed curve is the probability
oracle, the attainable minimum. The distance from the oracle curve, integrated with weight
$1/(c(1-c))$ and divided by $\ln 2$, equals the rule's betting oracle gap in bits
(Section~\ref{sec:schervish}).}
\label{fig:murphy}
\end{figure}

\subsection{Relation between the log score and the threshold losses}
\label{sec:schervish}

The betting oracle gap is a weighted integral of the excess costs over all decision
thresholds. Each threshold $c$ defines one binary decision problem, and a proper score
integrates these threshold losses with a specified weight measure. For the logarithmic score of the betting game the combination is:
\begin{equation*}
\text{betting oracle gap}
=\int_0^1 \bigl(\text{extra cost-loss expense at threshold }c\bigr)\,
\frac{\mathrm{d}c}{c(1-c)\,\ln 2}\,.
\end{equation*}
The general form is a representation theorem of
\citet{schervish1989}. The avoidable expense of a forecast $q$ on outcome $y\in\{0,1\}$ at
cost ratio $c$ is the elementary loss
\[
S_c(q,y)=
\begin{cases}
c & \text{if } q>c \text{ and } y=0 \quad\text{(unnecessary review)},\\
1-c & \text{if } q\le c \text{ and } y=1 \quad\text{(unreviewed fault)},\\
0 & \text{otherwise}.
\end{cases}
\]
Every sufficiently regular binary proper score has a representation
\begin{equation*}
S(q,y)=\int_0^1 S_c(q,y)\,\nu(\mathrm{d}c),
\end{equation*}
up to terms independent of $q$, for a nonnegative weight measure $\nu$; conversely, every
such mixture is a proper score. The measure $\nu$ determines the relative weight of
each threshold.\footnote{The weights are not probability densities, and the logarithmic
weight has no finite total.} Two examples are:
\begin{center}
\small
\begin{tabular}{ll}
\toprule
score & threshold weight $\nu(\mathrm{d}c)$\\
\midrule
logarithmic & $\mathrm{d}c/(c(1-c))$ --- more weight near 0 and 1\\
Brier & $2\,\mathrm{d}c$ --- equal weight over thresholds\\
\bottomrule
\end{tabular}
\end{center}
The logarithmic weight yields the displayed identity, given by \citet{ehm2016};
\citet{fissler2022} discuss its use in machine learning. The logarithmic weight $1/(c(1-c))$ grows without bound near $0$ and $1$, assigning
greater weight to decision errors at extreme thresholds. The Brier score assigns constant weight
over thresholds, so its rule rankings can differ.

The identity was checked numerically on the same rounds:
for every rule, the area between its expense curve and the oracle's, weighted by $1/(c(1-c))$
and divided by $\ln 2$ to convert natural logarithm to bits,
reproduces the rule's oracle gap from the betting run to four decimal places: for the
average in the two-models world, both calculations give $0.0891$ bits per round.

\subsection{The weighted and extremized log-odds family}
\label{sec:knob}

We vary the symmetric coefficient $\alpha$ of the family pooled-log-odds
$=w_1\ell_1+w_2\ell_2$ of Section~\ref{sec:geo} from $0.5$ to $2.5$ in four settings.
Both source weights equal $\alpha$, and the corresponding extremization factor is
$\gamma=2\alpha$.
\begin{enumerate}
\item Calibrated reports based on conditionally independent evidence: the oracle gap is
zero at $\alpha=1$ (upco; $\gamma=2$).
\item Both reports based on the same observation (the fully shared setting): the gap is
zero at $\alpha=1/2$ (geometric; $\gamma=1$, no extremization).
\item Underconfident reports, each model stating half its justified log-odds: the gap is
zero at $\alpha=2$ ($\gamma=4$), which gives logits twice the average reported logit. The
same fitted exponent is equivalent to temperature scaling applied to uniformly distorted
logits \citep{guo2017}.
\item Each report based on one shared and one source-specific observation (the
half-shared setting; fixed accuracies 0.8 and 0.75): the best symmetric $\alpha$ is 0.6,
and a gap of 0.044 bits per round remains at every $\alpha$. The required adjustment
depends on the observations of the individual round, which one common coefficient cannot
represent. One common log-odds coefficient exactly corrects uniform distortion and
complete duplication, but generally not partial sharing.
\end{enumerate}

Figure~\ref{fig:alpha} plots all four settings on one axis: the oracle gap as a function
of $\alpha$. The first three oracle-gap functions equal zero at their
predicted exponents (1, 1/2, 2); the half-shared setting has a positive minimum oracle gap
at an interior $\alpha$.

\begin{figure}[H]
\centering
\includegraphics[width=0.72\textwidth]{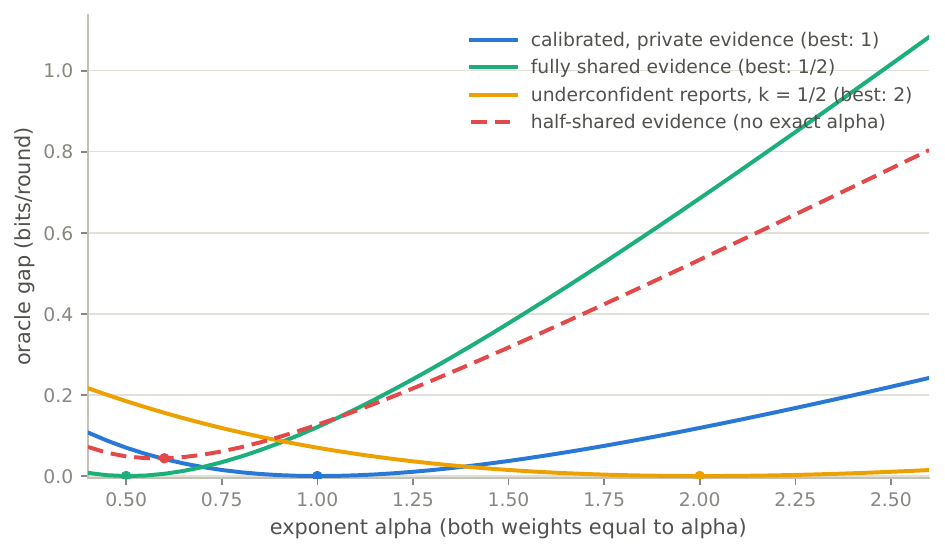}
\caption{The oracle gap (excess log loss) across the symmetric log-odds exponent $\alpha$
(the coefficient applied to each source's log-odds; extremization factor $\gamma=2\alpha$),
in four evidence settings. Lower is better; zero is oracle performance. Dots mark the best
$\alpha$ per setting. Calibrated private evidence,
fully shared evidence and uniformly underconfident reports reach the oracle at $\alpha=1$,
$1/2$ and $2$; half-shared evidence has an interior best $\alpha$ with a positive
minimum.}
\label{fig:alpha}
\end{figure}

Asymmetric distortion requires asymmetric weights: with model 1 doubling its log-odds and
model 2 halving its own, the weighted pool with $(w_1,w_2)=(1/2,\,2)$ attains the oracle
($+0.380$), while the best symmetric member reaches $+0.280$ and the average,
geometric, upco and noisy-or rules reach at most $+0.295$. On the probability scale, in the alternative-reading setting in which
the two cases have probabilities 0.7 and 0.3, the 0.7-weighted average attains the oracle and the
equal-weight average has an oracle gap of 0.022 bits per round.

These experiments treat the distortion factors, mixture weights, and overlap structure as
known inputs, as in worlds 8 and 10. Estimating these quantities from data, as
performance-weighting schemes do \citep{cooke1991}, is not studied here.

\subsection{Relation to existing evaluation practice}
\label{sec:evalpractice}

Forecast-combination studies evaluate aggregators by proper scores, calibration, and
discrimination \citep{ranjangneiting2010,satopaa2014,baron2014,jose2014}. The geometric
pools are externally Bayesian --- pooling and then conditioning on shared evidence
commute \citep{genest1984}; \citet{baccellistewart2023} analyse when geometric pooling
coincides with the update of a single Bayesian observer of the sources;
\citet{ranjangneiting2010} prove that nontrivial linear pools of distinct calibrated
forecasts are overdispersed under their stated conditions, and propose a recalibrating
transformation. Language-model
studies mainly compare voting and linear-pooling variants by expected calibration error
and failure-prediction AUROC \citep{xiong2024,lyu2024,verga2024}, and usually do not
model dependence among the reports. Betting-based evaluation has been used for individual
forecasters and forecast panels: Kelly-wealth comparisons
\citep{johnstone2007,johnstone2012}, betting returns against market odds
\citep{grantjohnstone2010}, simulated decision-making comparisons of inductive-logic
methodologies \citep{radzvilas2023}, and the weather-roulette presentation of forecast
value \citep{hagedornsmith2009}. The experiments compare probability pools, confidence
rules, Dempster--Shafer quantities, and proof unions with shared premises under stated
assumptions. Kelly growth represents the logarithmic score; the cost-loss curves show
how the same score weights decisions at different thresholds
\citep{schervish1989,gneitingraftery2007,ehm2016}.

\FloatBarrier
\section{Summary}
\label{sec:summary}

The paper addresses two questions about probability and confidence
combination.  First, what assumptions make a particular pooling rule correct?
Second, what is lost when a rule is used outside those assumptions?

The main results are:
\begin{enumerate}

\item The appropriate combination rule depends on both what the input numbers
represent and how their sources are related.  The paper gives explicit
generating conditions for rules that combine alternative cases, independent or
shared evidence, component events with different dependence relations,
supporting and opposing evidence, belief masses, and derivations with shared
premises.  The controlled Monte Carlo experiments reproduce the probabilities
derived for these settings.

For example, averaging is exact for alternative cases, prior-corrected odds
pooling for conditionally independent evidence with a common prior, and
maximum or probabilistic sum for different dependence relations between
component events.

\item Using a rule outside the conditions for which it is correct can
substantially increase both probability-sensitive loss and decision cost.
Several numerically different rules nevertheless make the same binary decision
at threshold $1/2$, so accuracy at that single threshold can hide important
differences between them.

\item Different ways in which supporting and opposing evidence interact also
lead to different combination rules.  Independent support and opposition give
the ProbLog-style construction; opposition that refutes particular supporting
items gives the revised MYCIN rule; and conditioning on the two sides not both
being decisive gives the Dempster--Shafer construction studied here.

\item When several proofs share uncertain premises, retaining the identities
of those premises makes it possible to calculate directly the probability
that at least one proof is available.  Reducing each proof to a single
probability loses information about shared premises and is insufficient in
general for combining three or more proofs.

\item Weighted log-odds pooling can exactly correct some simple forms of
systematic dependence or miscalibration, such as repeated evidence or a common
multiplicative distortion of the reported log-odds.  More general partial
sharing requires information about which evidence is shared.
\end{enumerate}

To choose a combination rule, three questions should be answered.

\begin{enumerate}
\item What do the inputs represent?  They may be posterior probabilities
about the same outcome, probabilities of separate component events, signed
strengths of evidence, belief masses, or proofs together with the uncertain
premises on which they depend.

\item How are the sources or events related?  For posterior reports, the
relevant distinction is whether they describe alternative cases, independent
evidence about the same case, or evidence that is shared or systematically
distorted.  For component events and proofs, the relevant question is whether
their occurrence is nested, independent, only partially dependent, or linked
through identified shared premises.

\item If evidence supports opposite sides of a claim, how does that
opposition operate?  The two sides may arise independently, opposing items may
refute particular supporting items, or the model may exclude cases in which
both sides are simultaneously decisive.
\end{enumerate}

These distinctions determine the applicable rules.  Alternative cases lead
to weighted averaging.  Conditionally independent evidence about the same
outcome leads to prior-corrected odds pooling; geometric or weighted
log-odds pooling is exact under the particular duplication and distortion
models studied here.  For component events, nested events give the maximum,
independent events the probabilistic sum, and the Fr\'echet bounds delimit
what can be inferred from marginals alone.  For overlapping proofs,
cumulation approximates dependence from an overlap parameter, whereas the GK
calculation uses identified shared premises directly.  For opposing evidence,
the constructions studied here give the ProbLog-style rule under independent
opposition, the revised MYCIN rule under item-level refutation, and
Dempster--Shafer quantities when simultaneous decisive support and opposition
are excluded.

Which of these assumptions approximately holds in a real application is an
empirical question.  Applying the rules in practice may therefore require
calibration data, estimates of dependence between source errors, and records
of which evidence is reused.

\paragraph{Limitations and future work.}
Most of the analysis concerns two reports and binary questions; extending the
results to larger collections of reports requires assumptions about their
joint dependence.  The experiments assume that the applicable generating
mechanism and its parameters are known; the paper does not study how to infer
them from data.  The rules are evaluated as isolated combination steps.
Future work should evaluate their use in complete prediction and reasoning
pipelines, including the calibration of propagated confidence values against
observed outcomes.

\section*{Acknowledgments}

This research was supported by the European Union and the Estonian Research Council
through project TEM-TA141, and the Estonian Centre of Excellence in Artificial Intelligence
(EXAI) project TK213U8, funded by the Estonian Ministry of Education and Research.

\section*{Declaration on generative AI}

During the preparation of this work the authors used generative-AI assistants ---
Claude Opus 4.8 and Claude Fable 5 (Anthropic) and GPT-5.6 Sol (OpenAI) ---
to draft experiment code, generate figures, edit text, and review drafts.
After using these tools, the authors reviewed and edited the content as needed
and take full responsibility for the content of this paper.

\bibliographystyle{plainnat}

\appendix

\section{Experiment code}

\begin{sloppypar}
Release \texttt{v2.0}, commit \texttt{40f2cac}, of
\url{https://github.com/tammet/poolingworlds} contains the complete experiment suite. The
default \texttt{run\_all.py} execution (about two minutes) reproduces the reported
results, and \texttt{expected\_output.txt} permits deterministic comparison:
\texttt{python3 run\_all.py | diff - expected\_output.txt}. Every experiment fixes its
random seeds; the tested environment was CPython~3.12.3, with NumPy~1.26.4 and
Matplotlib~3.6.3 for the figure scripts. The suite consists of small self-contained
Python files:
\end{sloppypar}
\begin{itemize}\itemsep 1pt \sloppy
\item \texttt{poolib.py} --- the shared pooling-function library;
\item \texttt{corr\_*.py} --- one file per Section~3 simulation, including
\texttt{corr\_gk\_measured\_overlap.py};
\item \texttt{decide\_kelly\_betting.py}, \texttt{decide\_kelly\_scenarios.py},
\texttt{decide\_kelly\_gk\_world.py} --- the betting harness and worlds (the GK
experiments are suite sections 7b and 10b);
\item \texttt{decide\_cost\_loss.py}, \texttt{decide\_murphy\_diagram.py} --- the
cost-loss protocol and Murphy diagram;
\item \texttt{decide\_extremized\_weighted\_logodds.py} --- the log-odds family sweep;
\item \texttt{run\_all.py} --- the runner; the \texttt{--extras} flag additionally runs
exploratory scripts not used by this paper.
\end{itemize}
\begin{sloppypar}
\noindent The GK reasoner itself is publicly usable through binaries, documentation,
examples and samplers at \url{https://github.com/tammet/gkreasoner} and through a browser
interface at \url{https://logictools.org/commonsense.html}. The central components are
reproduced below; function and file names match the repository.
\end{sloppypar}

\subsection{Pooling functions (from poolib.py)}

\begin{verbatim}
def avg(p, q):  return (p + q) / 2
def mx(p, q):   return max(p, q)
def por(x, y):  return x + y - x * y            # probabilistic sum / noisy-or

def upco(p, q):
    d = p * q + (1 - p) * (1 - q)
    return 0.5 if d == 0 else (p * q) / d

def geo(p, q, w=0.5):
    a = (p ** w) * (q ** w); b = ((1 - p) ** w) * ((1 - q) ** w)
    return 0.5 if a + b == 0 else a / (a + b)

def wlogodds(p, q, w1, w2):                     # weighted log-odds pool
    p = min(max(p, 1e-12), 1 - 1e-12); q = min(max(q, 1e-12), 1 - 1e-12)
    x = w1 * math.log(p / (1 - p)) + w2 * math.log(q / (1 - q))
    return 1 / (1 + math.exp(-x))

def bayes_odds_pool(ps, prior=0.5):             # odds pool with prior correction
    O0 = prior / (1 - prior)
    o = 1.0
    for p in ps: o *= p / (1 - p)
    o /= O0 ** (len(ps) - 1)
    return o / (1 + o)

def mycin_diff(x, y):                           # MYCIN conflict on strengths
    if x == y: return 0.0
    return (x - y) / (1 - y) if x > y else (y - x) / (1 - x)

# The revised MYCIN rule carried to the confidence scale [0,1], 0.5 = "no
# information", via strength = 2*conf - 1 (probtopneg below; ab = abs).
def probtopneg(x):  return x * 2 - 1
def pnegtoprob(x):  return (x + 1) / 2

def mycin(x, y):
    if (x == 0 and y == 1) or (x == 1 and y == 0): return 0.5
    x1 = abs(probtopneg(x)); y1 = abs(probtopneg(y))
    if x >= 0.5 and y >= 0.5:                   # agreement: noisy-or of strengths
        z = por(x1, y1)
    elif x <= 0.5 and y <= 0.5:
        z = -por(x1, y1)
    else:                                       # conflict: renormalized difference,
        z = mycin_diff(x1, y1)                  # sign of the stronger side
        if (x >= 0.5) != (x1 >= y1): z = -z
    return pnegtoprob(z)

# (Identical to the repository version, which spells out the 0/1 corner cases
# in explicit branches; the ProbLog conflict on strengths is a * (1 - b).)

def cumul(c1, c2, a):                           # CONFER cumulation
    return max(c1 + c2 * a, c1 * a + c2) - c1 * c2 * a
def gkpool(c1, c2, s):                     # GK two-proof union formula
    # two derivations, shared part visible with probability s (s=1: disjoint);
    # requires c1 <= s and c2 <= s (so s >= max(c1,c2));
    # s=1 -> noisy-or, s=max(c1,c2) (nested support) -> max
    return c1 + c2 - c1 * c2 / s

# In the ds_* functions the 0.5 returned at total conflict (m1*m2 = 1) is an
# experimental boundary convention; Dempster's rule itself is undefined there.
def ds_bel(m1, m2):
    n = 1 - m1 * m2
    return (m1 * (1 - m2)) / n if n > 0 else 0.5

def ds_pl(m1, m2):
    n = 1 - m1 * m2
    return 1 - (m2 * (1 - m1)) / n if n > 0 else 0.5

def ds_pignistic(m1, m2):
    n = 1 - m1 * m2
    if n <= 0: return 0.5
    return (m1 * (1 - m2) + 0.5 * (1 - m1) * (1 - m2)) / n
\end{verbatim}

The four-line ProbLog program for the support-with-veto construction of
Section~\ref{sec:problog}:

\begin{verbatim}
0.8::support.    0.3::oppose.
claim :- support, \+oppose.
query(claim).
\end{verbatim}

\noindent for which ProbLog returns $P(\texttt{claim})=0.8\cdot(1-0.3)=0.56=a(1-b)$.

\subsection{Betting harness and world generators (condensed)}

\begin{verbatim}
def growth_table(rounds, pools, prior):
    """rounds: list of (p_true, outcome, x1, x2);
    pools: name -> f(x1, x2, p_true) -> q.
    Stake split q : 1-q; the bookmaker pays 1/prior on the event."""
    g = {name: 0.0 for name in pools}
    for p_true, out, x1, x2 in rounds:
        for name, fn in pools.items():
            q = min(max(fn(x1, x2, p_true), 1e-6), 1 - 1e-6)
            g[name] += math.log2((q if out else 1 - q)
                                 / (prior if out else 1 - prior))
    return {name: g[name] / len(rounds) for name in pools}

def world_mixture(T):                      # world 1: exact rule = average
    rounds = []
    for _ in range(T):
        p1, p2 = random.choice([(0.30, 0.80), (0.15, 0.65),
                                (0.60, 0.90), (0.20, 0.50)])
        rate = p1 if random.random() < 0.5 else p2
        rounds.append(((p1 + p2) / 2, random.random() < rate, p1, p2))
    return rounds

def world_experts(T, prior=0.5):           # world 2: exact rule = odds pool
    rounds = []
    for _ in range(T):
        H = random.random() < prior
        rep = []
        for _ in range(2):
            acc = random.uniform(0.55, 0.95)
            sig = H if random.random() < acc else not H
            lh, ln = (acc, 1 - acc) if sig else (1 - acc, acc)
            rep.append(lh * prior / (lh * prior + ln * (1 - prior)))
        rounds.append((bayes_odds_pool(rep, prior), H, rep[0], rep[1]))
    return rounds

def poisson(lam):                          # Knuth's sampler (pure Python)
    L = math.exp(-lam); k, p = 0, 1.0
    while True:
        k += 1; p *= random.random()
        if p <= L: return k - 1

def world_refutation(T):                   # world 6: exact rule = MYCIN conflict
    rounds = []
    for _ in range(T):
        a = random.uniform(0.5, 0.95); b = random.uniform(0.05, a - 0.05)
        la, lb = -math.log(1 - a), -math.log(1 - b)
        n_for = poisson(la)                     # supporting signs
        survivors = sum(random.random() >= lb / la for _ in range(n_for))
        rounds.append(((a - b) / (1 - b), survivors >= 1, a, b))
    return rounds

def world_witnesses(T):                    # world 9: exact rule = DS pignistic
    rounds = []
    for _ in range(T):
        m1 = random.uniform(0.5, 0.95); m2 = random.uniform(0.05, m1 - 0.05)
        while True:
            r1, r2 = random.random() < m1, random.random() < m2
            if not (r1 and r2): break           # both decisive: redraw
        out = r1 or ((not r2) and random.random() < 0.5)
        rounds.append((ds_pignistic(m1, m2), out, m1, m2))
    return rounds
\end{verbatim}

The remaining world generators (the attempts world; doubled log-odds; independent
counter-evidence; all-or-nothing sharing) follow the same pattern and are in the
repository.

\subsection{A complete simulation check}

The following script is representative of the Section~\ref{sec:corr} experiments
(\texttt{corr\_averaging\_mixture.py}, complete except for print formatting): simulate the
mechanism, estimate the outcome frequency, and compare it with the candidate formulas.

\begin{verbatim}
def main(N=200_000, seed=1):
    random.seed(seed)
    pairs = [(0.30, 0.80), (0.15, 0.65), (0.20, 0.50),
             (0.60, 0.90), (0.40, 0.70)]
    for p, q in pairs:
        hits = 0
        for _ in range(N):
            rate = p if random.random() < 0.5 else q   # selected case rate
            if random.random() < rate:                 # sampled outcome
                hits += 1
        freq = hits / N
        # compare freq against avg, max, noisy-or, upco, geo on (p, q)
        # and report the best-matching rule; expected: avg, within +-0.002
        print(poolib.fmt_compare(freq, p, q))
\end{verbatim}

\subsection{Cost-loss components}

The following function implements the elementary cost-loss value of
Section~\ref{sec:costloss}, stated formally in Section~\ref{sec:schervish}; the final
comment states its numerical relation to the betting oracle gap.

\begin{verbatim}
def elementary_loss(q, y, c):        # forecast q, outcome y, cost ratio c
    act = q > c
    if y and not act:  return 1 - c  # unreviewed fault
    if act and not y:  return c      # unnecessary review
    return 0.0

# Schervish check: for each rule, mean elementary loss minus the oracle's,
# integrated over c with weight 1/(c*(1-c)), divided by ln 2, equals the
# rule's oracle gap from growth_table -- verified to four decimals.
\end{verbatim}

\end{document}